\documentclass[12pt]{article}
\pdfoutput=1
\usepackage{jheppub}
\usepackage{mathtools,amsmath,amssymb,amsthm,physics,bm,float}
\usepackage[utf8]{inputenc}
\usepackage{bbm}
\usepackage{graphicx}
\usepackage{enumitem}
\usepackage[normalem]{ulem}
\usepackage{hyperref}
\usepackage{comment}

\theoremstyle{definition}

\newtheorem{theorem}{Theorem}
\newtheorem{lemma}[theorem]{Lemma}

\theoremstyle{remark}

\newcommand{\be}{\begin{equation}}
\newcommand{\ee}{\end{equation}}
\newcommand{\bi}{\begin{itemize}}
\newcommand{\ei}{\end{itemize}}
\def\ba#1\ea{\begin{align}#1\end{align}}
\def\bg#1\eg{\begin{gather}#1\end{gather}}
\def\bm#1\em{\begin{multline}#1\end{multline}}
\def\bmd#1\emd{\begin{multlined}#1\end{multlined}}

\def\a{\alpha}
\def\b{\beta}
\def\c{\chi}

\def\d{\delta}
\def\D{\Delta}
\def\e{\epsilon}

\def\g{\gamma}

\def\l{\lambda}
\def\L{\Lambda}
\def\m{\mu}

\def\r{\rho}
\def\s{\sigma}

\def\y{\psi}

\newcommand{\la}{\label}

\newcommand{\re}{\ref}
\newcommand{\er}{\eqref}

\newcommand{\fr}{\frac}

\newcommand{\td}{\widetilde}

\newcommand{\eq}{\equiv}

\newcommand{\cd}{\cdots}
\newcommand{\nn}{\nonumber}
\newcommand{\qu}{\quad}
\newcommand{\qqu}{\qquad}
\newcommand{\lt}{\left}
\newcommand{\rt}{\right}

\newcommand{\ol}{\overline}

\renewcommand{\(}{\left(}
\renewcommand{\)}{\right)}
\renewcommand{\[}{\left[}
\renewcommand{\]}{\right]}
\newcommand{\<}{\langle}
\renewcommand{\>}{\rangle}

\newcommand{\bC}{{\mathbb C}}

\newcommand{\bR}{{\mathbb R}}

\newcommand{\cA}{{\mathcal A}}
\newcommand{\cB}{{\mathcal B}}

\newcommand{\cH}{{\mathcal H}}

\newcommand{\cO}{{\mathcal O}}
\newcommand{\cZ}{{\mathcal Z}}
\newcommand{\Bb}{{\ol B}}
\newcommand{\ob}{{\bar o}}
\newcommand{\rb}{{\bar r}}
\newcommand{\ub}{{\bar u}}

\def\tr{\operatorname{tr}}
\def\Tr{\operatorname{Tr}}

\newcommand{\bdy}{\text{bdy}}

\def\UV{\text{UV}}
\def\IR{\text{IR}}
\def\pre{\text{pre}}
\def\seed{\text{seed}}
\def\rest{\text{rest}}
\def\supp{\operatorname{supp}}
\def\vN{\operatorname{vN}}

\def\ln{\log}

\begin{document}

%% XD = Xi Dong
\definecolor{acolor}{rgb}{0.1,.6,0}
\newcommand{\XD}[1]{{\color{acolor}#1}}
\newcommand{\XDc}[1]{\textcolor{acolor}{[XD:#1]}}
%% DM  =   Don Marolf
\definecolor{rcolor}{rgb}{0.9,0.1,0.1}
\newcommand{\DM}[1]{{\color{rcolor}#1}}
\newcommand{\DMc}[1]{\textcolor{rcolor}{[DM:#1]}}
\definecolor{mcolor}{rgb}{0.8,0.1,0.6}
\newcommand{\DMTD}[1]{\textcolor{mcolor}{[DM:#1]}}
%% PR  =   Pratik Rath
\definecolor{bcolor}{rgb}{0.1,0,1}
\newcommand{\PR}[1]{{\color{bcolor}#1}}
\newcommand{\PRc}[1]{\textcolor{bcolor}{[PR:#1]}}

\title{Holographic codes and bulk RG flows}

\author[a]{Xi Dong,}
\author[a]{Donald Marolf,}
\author[b]{and Pratik Rath}

\affiliation[a]{Department of Physics, University of California, Santa Barbara, CA 93106, USA}
\affiliation[b]{Leinweber Institute for Theoretical Physics and Department of Physics,\\
University of California, Berkeley, California 94720, U.S.A.}

\emailAdd{xidong@ucsb.edu}
\emailAdd{marolf@ucsb.edu}
\emailAdd{pratik\_rath@berkeley.edu}

\abstract{We consider the coarse-graining of holographic quantum error correcting codes under a generalized notion of bulk renormalization-group flow.   In particular, we study the renormalization under this flow of the $A/4G$ term in the Faulkner-Lewkowycz-Maldacena formula and in its R\'enyi generalization.  This provides a general quantum code perspective on the arguments of Susskind and Uglum. Specifically,  given a `UV' code with two-sided recovery and appropriately flat entanglement spectrum together with a set of `seed' states in the UV code, we explicitly construct an `IR' code with corresponding properties which contains the given seed states and is of minimal size in a sense we describe.
}

\maketitle

%-------------------------------------------------
\section{Introduction}

It is now well-established that an appropriate semiclassical bulk limit of the AdS/CFT bulk-to-boundary map can be thought of as a quantum error correcting code \cite{Almheiri:2014lwa}; see e.g. \cite{Jafferis:2015del,Dong:2016eik,Harlow:2016vwg,Kamal:2019skn}.  This idea has thus become an integral part of modern discussions of holography, and in fact forms the basis of certain proposals \cite{Akers:2022qdl,Harlow:2025pvj} for understanding the physics of black hole interiors and closed cosmologies.  

Here we explore the behavior of holographic codes under renormalization-group-like transformations in which a bulk Hilbert space is first defined at some ultraviolet (UV) scale and then truncated to a smaller Hilbert space that might, for example, allow rather general infrared (IR) behavior while requiring the UV degrees of freedom to be in a local vacuum state. Susskind and Uglum \cite{Susskind:1994sm} proposed that changes in renormalization-group scale are associated with changes in both Newton's constant $G$ and the von Neumann entropy $S_{vN}$ of certain subsystems such that the generalized entropy $S_{gen} = \frac{A}{4G} + S_{vN}$ remains invariant, where $A$ is the area of an appropriate surface bounding the desired subsystem. Related issues regarding quantum error correction in AdS/CFT were recently explored in \cite{Gesteau:2023hbq}, which demonstrated a setup where the boundary entropy computed using the Faulkner-Lewkowycz-Maldacena (FLM) formula \cite{Faulkner:2013ana} is invariant under renormalization-group (RG) flow.\footnote{See also \cite{Furuya:2020tzv} for discussion of the renormalization group as an error correcting code in contexts without complementary recovery.}

Our analysis will go beyond \cite{Gesteau:2023hbq} in several ways. The first is that we will analyze settings where the bulk algebras recovered by our codes have non-trivial centers in both the UV and the IR. This allows us to treat the area $A$ in the FLM formula as an operator as is natural in gravitational theories \cite{Faulkner:2013ana,Harlow:2016vwg}. In general, we will find these centers to change significantly under our RG flow. Secondly, a special feature of holographic codes is that they possess an (approximately) flat entanglement spectrum so that they satisfy a R\'enyi generalization of the FLM formula \cite{Dong:2018seb,Akers:2018fow,Dong:2019piw}. Our RG flow is designed to ensure that, when this property holds in the UV, it also remains true in the IR. Finally, whereas \cite{Gesteau:2023hbq} assumed the existence of a sequence of codes that are related by RG flow, we will explicitly construct such codes given a set of UV `seed' states that are required to remain in the IR Hilbert space.

Since we will largely work from a code perspective, it is natural to consider very general contexts in which the coarse-graining that defines the IR degrees of freedom need not be strictly local.  As a result, we will think of the resulting flow as a renormalization of the full geometric entropy operator $\sigma = \frac{A}{4G} +\cd$ (where dots include higher-derivative corrections), rather than simply renormalizing the local coupling $G$.  

The above goals are motivated in part because we wish to use our RG transformation in forthcoming works \cite{JLMS,ModFlow} to study modular flow while avoiding the kind of failures of the Jafferis-Lewkowycz-Maldacena-Suh (JLMS) formula \cite{Jafferis:2015del} found in \cite{Kudler-Flam:2022jwd}. As will be explained in \cite{JLMS}, such failures are fundamentally associated with small eigenvalues for the density matrix of a bulk subregion.  Because UV Hilbert spaces tend to be very large, small eigenvalues are difficult to avoid in a UV description. Flowing to an IR description in terms of a smaller Hilbert space will thus make such issues easier to control. 

We will treat the Hilbert space $\cH_\UV$ of our UV holographic code as being finite-dimensional. This means that we impose both UV and IR cutoffs in the bulk, and that we also impose a cutoff on the amplitude of any (bosonic) excitation about some reference state (perhaps a classical background). We will use the term `holographic code' to mean a quantum error-correcting code with complementary recovery (which we will also sometimes call `two-sided recovery') \cite{Harlow:2016vwg} and flat entanglement spectrum in the sense of \cite{Dong:2018seb,Akers:2018fow,Dong:2019piw}.
For simplicity, up until section \re{sec:disc} we will assume these two properties to be exact in our UV holographic code.
Of course, at finite $G$ the AdS/CFT dictionary is known to only approximately have these two properties. But it will simplify our presentation to postpone explicit treatment of such approximations until the final discussion in section \ref{sec:disc}. By that time it will be clear that our treatment of the exact case carries over verbatim.  This is because our main construction takes as input only the UV bulk Hilbert space ${\cal H}_{\UV}$ (and a set of `seed' states therein) and then constructs an IR bulk Hilbert space $\cH_\IR$ whose embedding into $\cH_\UV$ satisfies exact two-sided recovery. As a result, the IR bulk-to-boundary map simply inherits any small deviations from exact two-sided recovery in the UV bulk-to-boundary map. Similar statements will apply to approximate flatness of the entanglement spectrum.

We begin in section \ref{sec:TechOverview} with a more technical overview of our results and an outline of our methods.  We then describe our construction in the particularly simple case in which we specify only a single seed state.  This is done in section \ref{sec:onestate}.  As we will see, choosing this state to be the only state in ${\cal H}_{\IR}$ does not generally lead to a code with an approximate flat entanglement spectrum; additional states will typically need to be added as well, which we will do by chopping the seed state into pieces defined by eigenvalue windows of its boundary modular Hamiltonian.

We then proceed in section \ref{sec:general} to the general construction where we choose an arbitrary set of seed states to be included in the IR code. We will first construct the smallest possible space that contains the given seed states and whose embedding into ${\cal H}_{\UV}$ defines a code with perfect two-sided recovery.  The states in that space will again be further chopped into pieces defined by appropriate eigenvalue windows, constructing our desired IR code with a degree of flatness for its entanglement spectrum set by the sizes of these windows (and which becomes very flat when the windows become small).

Section \ref{sec:ex} provides illustrative examples showing both how our construction can lead to a small ${\cal H}_{\IR}$ and how, for certain choices of seed states, no coarse-graining of the UV code will preserve exact two-sided recovery.  In particular, in the second example the smallest allowed IR code will in fact have ${\cal H}_{\IR} ={\cal H}_{\UV}$.
We then close with some final discussion in section \ref{sec:disc}.

\subsection{Results and methods}
\label{sec:TechOverview}

As noted above, our goal is to understand how the structure of holographic codes evolves under bulk RG flows. In particular, we begin with a holographic code which we think of as being defined at some UV scale $\L_\UV$ in the bulk.  As mentioned earlier, we use the term `holographic code' to mean a quantum error correcting code with (for now, exact) two-sided recovery and flat entanglement spectrum.  The corresponding code space $\cH_\UV$ may be thought of as consisting of states in (an appropriate limit of) the bulk effective field theory at the scale $\L_\UV$.  We in principle allow this Hilbert space to contain states describing different semiclassical geometries.  Consistent with \cite{Harlow:2016vwg,Kamal:2019skn}, we take the UV Hilbert space to be written as a sum of tensor products
\be\la{huv}
\cH_\UV = \bigoplus_\a \cH_u^\a \otimes \cH_\ub^\a,
\ee
where $u$, $\ub$ correspond to the entanglement wedges of two complementary subregions $B$, $\Bb$ on the boundary.  More precisely, the algebras $\cA_u$, $\cA_\ub$ of operators in the entanglement wedges of $B$, $\Bb$ that act on $\cH_\UV$ are commutants of each other and define the decomposition \er{huv} with $\cA_u = \bigoplus_\a \cB(\cH_u^\a)$, $\cA_\ub = \bigoplus_\a \cB(\cH_\ub^\a)$, where $\cB(\cH_u^\a)$, $\cB(\cH_\ub^\a)$ are the algebras of all bounded operators on $\cH_u^\a$, $\cH_\ub^\a$, respectively. We require these Hilbert spaces to have finite dimension which, from a bulk point of view, means that they are defined by imposing both UV and IR cutoffs (and also cutoffs on the amplitude of bosonic excitations).

Since the code satisfies exact two-sided recovery, it follows from \cite{Harlow:2016vwg} that states in $\cH_\UV$ satisfy the quantum-corrected Ryu-Takayanagi (RT) formula, also known as the Faulkner-Lewkowycz-Maldacena formula (or QES formula) \cite{Faulkner:2013ana}
\be\label{flmuv}
S(\td\r_B) = \<A_\UV\>_\r + S(\r_u),
\ee
where $\r$ is a state (i.e., a density operator) on $\cH_\UV$ and $\r_u$ is the corresponding reduced state with respect to $\cA_u$ defined by the requirements that it is in $\cA_u$ (and may therefore be represented as a block diagonal matrix acting on $\cH_u:= \bigoplus_\a \cH_u^\a$) and reproduces the expectation values in $\r$ of all operators in $\cA_u$.
The notation $\td\r$ denotes the boundary CFT state encoded by $\rho$, and $\td\r_B$ is the corresponding reduced state on the boundary subregion $B$. Moreover, $A_\UV$ is an appropriate `geometric entropy' operator on ${\cal H}_{\UV}$ that lies in the center $\cZ_u$ of $\cA_u$. In Einstein gravity this $A_\UV$ would be $A/4G$ where $A$ is the geometric area.\footnote{This $A_\UV$ was called $\mathcal L$ in \cite{Harlow:2016vwg,Kamal:2019skn}, and in higher-derivative gravity it is sometimes called $\sigma$ \cite{Dong:2019piw,Dong:2025orj}.}  For simplicity, we will nevertheless refer to $A_\UV$ as the UV area operator below, using the terminology of Einstein gravity with the convention $G=1/4$.

The expectation value $\<A_\UV\>_\r$ in \er{flmuv} is defined as $\tr(\r A_\UV)$. The bulk von Neumann entropy $S(\r_u)$ is defined by $S(\r_u):=-\tr_u \left(\r_u \ln \r_u \right)$ with the trace defined by summing expectation values over an orthonormal basis in $\cH_u := \bigoplus_\a \cH_u^\a$.  This is a particular choice of trace (and the definition of entropy) on the UV algebra $\cA_u$, and any other choice differs only by a term that may be absorbed into the definition of $A_{\UV}$.  We will thus call \er{flmuv} the FLM formula in the UV.

Our goal will be to identify suitable subspaces of bulk states $\cH_\IR \subset \cH_\UV$  which will define our IR code.  In appropriate contexts one may think of $\cH_\IR$ as consisting of states in the bulk effective field theory at some IR scale $\L_\IR$. The main result of this paper is to give a general, explicit way of constructing such an $\cH_\IR$ so that it again satisfies an FLM formula (and its R\'enyi version) in the IR subject to the constraint that it contains a given set of UV states. In particular, our construction of $\cH_\IR$ as a subspace of $\cH_\UV$ will naturally define IR algebras $\cA_r$, $\cA_\rb$ of operators in the entanglement wedges of the boundary regions $B$, $\Bb$ that act within $\cH_\IR$.\footnote{We chose the names $r$, $\rb$ to suggest the IR, the names $u$, $\ub$ to suggest the UV, and the names $B$, $\Bb$ to suggest the boundary.} These algebras will be commutants of each other in $\cB(\cH_\IR)$ and will thus define a decomposition of $\cH_\IR$ analogous to \er{huv}:
\be\la{hir}
\cH_\IR = \bigoplus_\m \cH_r^\m \otimes \cH_\rb^\m,
\ee
with $\m$ labeling different IR superselection sectors,\footnote{We will construct these superselection sectors explicitly.  The IR superselection sectors $\m$ will generally differ from the UV superselection sectors $\a$. Indeed, as we will see, we can have multiple $\m$-sectors arise from a single $\a$-sector, or a set of $\a$-sectors can collapse to a single $\m$-sector.} and with $\cA_r = \bigoplus_\m \cB(\cH_r^\m)$, $\cA_\rb = \bigoplus_\m \cB(\cH_\rb^\m)$. Moreover, we will show that states $\r$ on $\cH_\IR$ satisfy an IR FLM formula
\be\la{flmir}
S(\td\r_B) = \<A_\IR\>_\r + S(\r_r),
\ee
with an explicitly constructed IR area operator $A_\IR$ in the center $\cZ_r$ of $\cA_r$. Here $\r_r$ is the IR reduced density operator defined with respect to $\cA_r$ and $S(\rho_r):= - \tr_r \left( \rho_r \ln \rho_r\right)$ is computed using the Hilbert space trace on $\cH_r := \bigoplus_\mu \cH_r^\mu$. As usual, we could alternatively use any other trace on $\cA_r$ by absorbing the difference into the definition of $A_{\IR}$.

Our strategy for deriving the IR FLM formula \er{flmir} is to show that the change of the bulk entropy under the bulk RG flow takes the form
\be
S(\r_u) = S(\r_r) + \<\D A\>_\r,
\ee
for any IR state $\r$ and some operator $\D A$. Such a change naturally `renormalizes' the IR area operator in the sense that we may simply define
\be
A_\IR = \td A_\UV + \D A, 
\ee
where $\td A_\UV$ acts on $\cH_\IR$ by first applying $A_\UV$ and then projecting back onto $\cH_\IR$. The IR FLM formula is then equivalent to the UV FLM formula for the particular states in the IR code. Moreover, we will show that a certain notion of `subregion orthogonality' holds between different $\m$-sectors. This property will prove useful in \cite{JLMS}.

In order for our IR code to have an approximately flat entanglement spectrum in the sense of
\cite{Dong:2018seb,Akers:2018fow,Dong:2019piw},\footnote{\la{codeK}For a code with two-sided recovery, the bulk-to-boundary map may be viewed (in each bulk superselection sector) as taking a tensor product with some fixed $|\td\c\>$ state that lives in part of the boundary Hilbert space $\cH_\bdy$. The code has an approximately flat entanglement spectrum if the modular Hamiltonian of this $|\td\c\>$ is approximately a c-number. Since this modular Hamiltonian is fixed for a chosen code (and a superselection sector therein) and does not depend on the bulk state, we will call it the `code modular Hamiltonian.' Note that this $|\td\c\>$ is different from the $|\c\>$ states that are analyzed in sections \re{sec:onestate}-\re{sec:general} and which specify the IR-to-UV map.} our construction will also need to satisfy an approximate FLM-like formula for R\'enyi entropies in the IR.\footnote{Such a R\'enyi FLM formula holds exactly if and only if the code has a precisely flat entanglement spectrum \cite{Akers:2018fow}.} We will see that this is the case when such a formula holds in the original UV theory. This R\'enyi FLM formula is simplest to present for states $\r^\m$ in a single $\m$-sector of $\cH_\IR$ (which are analogous to fixed-area states \cite{Akers:2018fow,Dong:2018seb} in the IR). We will show that such states approximately satisfy
\be\la{rflmfixed}
S_n(\td\r_B^\m) = \<A_\IR\>_{\r^\m} + S_n(\r_r^\m),
\ee
for all $n\in{\mathbb C}$ such that the UV version holds,
and that more general IR states satisfy a slightly more complicated version \er{rflmir}, in agreement with the gravitational predictions of \cite{Akers:2018fow,Dong:2018seb}.

%-------------------------------------------------
\section{A simple example of the IR code}
\label{sec:onestate}

Here we consider a simple example where we choose only a single seed $|\c\>$ that we require to be included in $\cH_\IR$.  
One may think of this example as modeling the case where we follow the bulk RG flow to the deep IR, so that almost all quantum fluctuations have been integrated out. As stated in the introduction, here and in section \ref{sec:general} we take the UV code to have exact two-sided recovery and an exactly-flat entanglement spectrum (so that the R\'enyi FLM formula holds exactly).  After doing so, it will be straightforward to return in section \ref{sec:disc} to cases where these properties hold only approximately in the UV and to incorporate the effects of such UV errors into our IR error bounds.

Our goal is to find what is, in some sense, the `smallest' $\cH_\IR$ that contains $|\c\>$ and satisfies an exact FLM formula \er{flmir} as well as an approximate version of its R\'enyi generalizations \er{rflmfixed}. It is worth noting that, in general, we cannot choose $\cH_\IR$ to be the one-dimensional Hilbert space containing only $|\c\>$. If we were to do so, we could still easily satisfy the FLM formula \er{flmir} by setting $\D A$ to the c-number $S(\c_u)$, where $\c_u$ is the reduced density operator of $|\c\>$.  However, we would then generally have difficulty satisfying the R\'enyi FLM formula \er{rflmfixed} since the boundary entanglement spectrum of $|\c\>$ generically fails to be approximately flat.

To proceed, we note that, as a state in $\cH_\UV = \bigoplus_\a \cH_u^\a \otimes \cH_\ub^\a$, our $|\c\>$ has a Schmidt decomposition of the form
\be\la{cschm}
|\c\> = \bigoplus_\a \sum_i c_{\a i} |i\>_u^\a |i\>_\ub^\a,
\ee
where $c_{\a i}\neq 0$ by convention (i.e., we keep only nonzero terms in the sum), and $\{|i\>_u^\a\}$, $\{|i\>_\ub^\a\}$ are orthonormal sets in $\cH_u^\a$, $\cH_\ub^\a$, respectively. The corresponding reduced density operator $\c_u$ in $\cA_u$ is
\be
\c_u = \bigoplus_\a \sum_i |c_{\a i}|^2 |i\>_u^\a \<i|_u^\a,
\ee
and we define the associated UV bulk modular Hamiltonian to be $K_{\c_u} := -\log \c_u$. On the other hand, since $A_\UV$ is an operator in the center $\cZ_u$ of $\cA_u$, it must be of the form
\be\la{auv}
A_\UV = \bigoplus_\a A_\UV^\a \mathbbm 1^\a,
\ee
where $A_\UV^\a$ is a c-number and $\mathbbm 1^\a$ is the identity operator on the $\a$-sector. Therefore, on the support of $\c_u$ (viewed as a subspace of $\cH_\UV$), we have
\be
\label{eq:lambdas}
A_\UV + K_{\c_u} = \bigoplus_\a \sum_i \l_{\a i} |i\>_u^\a \<i|_u^\a,
\ee
where 
\be\la{lidef}
\l_{\a i} := A_\UV^\a -\log |c_{\a i}|^2.
\ee
A helpful way of understanding these $\l_{\a i}$ is that they are the eigenvalues of the boundary modular Hamiltonian (on the region $B$) corresponding to the $|\c\>$ state.

We now chop the state $|\c\>$ into pieces by grouping these eigenvalues $\l_{\a i}$ into `bins' of some small width. In other words, we divide the entire range of $\l_{\a i}$ into intervals labeled by some index $\g$. Each bin corresponds to an interval $I_\g = [\l_\g-\e_\g/2, \l_\g+\e_\g/2)$ for the eigenvalues, and so defines a truncated state in which we  keep only those terms in \er{cschm} with $\l_{\a i}$ in the interval $I_\g$:
\be\la{cgdef}
|\c^\g\> := \fr{1}{N_\g} \bigoplus_\a \sum_{i: \l_{\a i} \in I_\g} c_{\a i} |i\>_u^\a |i\>_\ub^\a,
\ee
where $N_\g$ is a normalization constant satisfying
\be\la{ngdef}
|N_\g|^2 = \sum_\a \sum_{i: \l_{\a i} \in I_\g} |c_{\a i}|^2.
\ee
We then define the IR Hilbert space as the linear span of all the truncated states resulting from our chopping procedure:
\be
\cH_\IR = \operatorname{span} \{|\c^\g\>: \g \in J\},
\ee
where $J$ is the set of all indices $\g$. In the rest of the paper, sums over $\g$ are implicitly taken over all elements of $J$.

Let us verify that this $\cH_\IR$ achieves the goals described in the overview (section \ref{sec:TechOverview}). First, it clearly has the decomposition $\cH_\IR = \bigoplus_\m \cH_r^\m \otimes \cH_\rb^\m$ required by \er{hir}: we simply identify the index $\m$ with $\g$,
\be
\m \eq \g,
\ee
and define each $\cH_r^\m \otimes \cH_\rb^\m$ to be the 1-dimensional Hilbert space containing $|\c^\m\>$. The IR algebras $\cA_r = \bigoplus_\m \cB(\cH_r^\m)$, $\cA_\rb = \bigoplus_\m \cB(\cH_\rb^\m)$ are both equal to the center $\cZ_r$. All three algebras are generated by the (commuting) projections $\mathbbm 1^\m = |\c^\m\> \<\c^\m|$.

\subsection{FLM formula}

It remains to show that the FLM and R\'enyi FLM formulas hold in the IR with some appropriate $A_\IR$ and to some accuracy. We begin with the non-R\'enyi case, choosing to define our $A_\IR$ so that it holds exactly (when \eqref{flmuv} is exact in the UV).

Let us first consider the simple case where there is only one $\a$-sector in the UV. 
Since $\cH_r^\m$ is one-dimensional, the corresponding IR reduced density matrix $\rho_r^\mu$ is the identity (no additional normalization coefficient is required) and so has vanishing entropy.  For a state in a given $\m$-sector, the difference between the von Neumann entropies on $\cH_r^\m$ and $\cH_u$ is thus precisely the UV entanglement entropy $S(\chi_u^\gamma) =-\tr_u\left(\chi_u^\g \ln \chi_u^\g \right)$ of $|\c^\g\>$ (recall $\g=\m$). 
We may thus compensate for the change of the bulk entropy by `renormalizing' the area operator via
\be\la{aironea}
A_\IR := \bigoplus_\g \left[A_\UV^\a + S(\chi_u^\gamma)\right]{\mathbbm 1}^\gamma,
\ee
where $\a$ takes a unique value in the special case considered here and where ${\mathbbm 1}^\gamma$ is the identity on each one-dimensional $\gamma$-sector.

In the more general case with multiple $\a$-sectors in the UV, since we still have only one state $|\chi^\gamma\rangle$ in each $\gamma$-sector, we simply define
\be\la{airdef}
A_\IR := \bigoplus_\g A_\IR^\g {\mathbbm 1}^\gamma,\qu
A_\IR^\g := \langle \chi^\gamma |A_\UV|\chi^\gamma\rangle + S(\chi_u^\gamma)
\ee
Since a general IR state takes the form $|\psi\rangle = \sum \psi_\gamma |\chi^\gamma\rangle$, the corresponding density matrix on $\cH_r := \bigoplus_\gamma \cH^\gamma_r$ may be written
\begin{equation}
\rho_r = \bigoplus_\g p_\gamma {\mathbbm 1}^\gamma   
\end{equation}
with $p_\gamma = |\psi_\gamma|^2$.  Similarly, noting 
from \er{cgdef} that for all $a \in \cA_u$ we have
\begin{equation}
\langle \psi | a | \psi \rangle = \sum_\gamma p_\gamma \langle \chi^\g | a | \chi^\g \rangle,
\end{equation}
it follows that the UV density matrix $\rho_u$ takes the form
\begin{equation}
\label{eq:1ru}
\rho_u = \bigoplus_\g p_\g \chi_u^\g,
\end{equation}
where $\bigoplus$ is used instead of $\sum$ to emphasize that the Hermitian operators being summed over have non-overlapping support.
A short calculation then yields
\begin{equation}
\langle A_\UV\rangle_{\rho} +  S(\rho_u) = \langle A_\IR\rangle_{\rho} + S(\r_r)
\end{equation}
for any state $\r$ in $\cH_\IR$.
Combining this with \eqref{flmuv} then shows that \eqref{flmir} holds as desired.

\subsection{R\'enyi FLM formula}
\label{sec:RenyiFLM}

We now turn to our R\'enyi generalization of \eqref{flmir}.  From \cite{Akers:2018fow,Dong:2018seb}, we see that the correct generalization    in the UV is
\be\la{rflmuv}
e^{(1-n)S_n(\td\r_B)} = \sum_\a p_\a^n e^{(1-n)\[A_\UV^\a+ S_n(\r_u^\a)\]},
\ee
where $p_\a$, $\r_u^\a$ are defined by $\r_u =\bigoplus_\a p_\a \r_u^\a$ and $\tr \r_u^\a=1$. In the special case of a state $\r^\a$ in a single $\a$-sector (i.e., a fixed-area state in the UV), \er{rflmuv} reduces to the simple statement
\be
S_n(\td\r_B^\a) = A_\UV^\a+ S_n(\r_u^\a).
\ee

We now show that if the UV R\'enyi FLM formula \er{rflmuv} holds, the IR version
\be\la{rflmir}
e^{(1-n)S_n(\td\r_B)} = \sum_\m p_\m^n e^{(1-n)\[A_\IR^\m+ S_n(\r_r^\m)\]}
\ee
also holds approximately, where $p_\m$, $\r_r^\m$ are similarly defined by $\r_r =\bigoplus_\m p_\m \r_r^\m$ and $\tr \r_r^\m=1$ and where the error vanishes in the limit $\e_\g \rightarrow 0$. This will then immediately imply the simple statement \er{rflmfixed} for a state that lies in a single $\m$-sector.

To show the IR R\'enyi FLM formula \er{rflmir} from the UV version \er{rflmuv}, we aim to show that for any state $\r$ in $\cH_\IR$, we have
\begin{equation} 
\label{eq:1target}
\sum_\a p_\a^n e^{(1-n)\[A_\UV^\a+ S_n(\r_u^\a)\]} \approx \sum_\m p_\m^n e^{(1-n)\[A_\IR^\m+ S_n(\r_r^\m)\]}.
\end{equation}
Recalling that $S_n(\r) = \fr{1}{1-n}\log \tr \r^n$, we find \eqref{eq:1target} to be equivalent to
\begin{equation}
\la{Renyieq}
\sum_\a e^{(1-n) A_\UV^\a} \tr_u \(p_\a\r_u^\a\)^n \approx \sum_\m e^{(1-n) A_\IR^\m} \tr_r \(p_\m\r_r^\m\)^n.
\end{equation}

To establish this relation, we again write $\r_u$ as in \eqref{eq:1ru}.  It is then useful to observe that, since $\chi_u^\gamma$ is defined by requiring that
\begin{equation}
\langle \chi^\gamma|a |\chi^\gamma\rangle = \tr_u \( a \chi_u^\gamma \)
\end{equation}
for all $a\in \cA_u$, we may use \eqref{cgdef} to check that we have
\be
\label{eq:chimg}
\c^\g_u := \tr_\ub |\c^\g\>\<\c^{\g}| = \bigoplus_\a \sum_{i: \l_{\a i} \in I_\g} \fr{|c_{\a i}|^2}{|N_\g|^2} |i\>_u^\a \<i|_u^\a.
\ee
One may think of the relation \eqref{eq:chimg} as being obtained by tracing $|\chi^\g\rangle \langle\chi^\g|$ over the Hilbert space $\cH_{\bar u}:= \bigoplus_\alpha \cH^\a_{\bar u}$.
We may then use this result to write \eqref{eq:1ru} as
\begin{eqnarray}
\label{eq:1ru2}
\rho_u = \bigoplus_\gamma p_\gamma \chi_u^\gamma = \bigoplus_\g \bigoplus_\a   \sum_{i: \l_{\a i} \in I_\g} p_\g \fr{|c_{\a i}|^2}{|N_\g|^2} |i\>_u^\a \<i|_u^\a, 
\end{eqnarray}
Decomposing \eqref{eq:1ru2} into $\alpha$-blocks then yields
\begin{eqnarray}
\label{eq:1ru3}
p_\a \rho_u^\a = \sum_\g \sum_{i: \l_{\a i} \in I_\g} p_\g \fr{|c_{\a i}|^2}{|N_\g|^2} |i\>_u^\a \<i|_u^\a. 
\end{eqnarray}

Using \er{lidef} we find the following useful approximation for $A_\IR^\g$ defined in \er{airdef}:
\ba
A_\IR^\g &= \sum_\a \sum_{i: \l_{\a i} \in I_\g} \fr{|c_{\a i}|^2}{|N_\g|^2} \(A_\UV^\a - \log \fr{|c_{\a i}|^2}{|N_\g|^2}\) \la{airgexact}\\
&= \sum_\a \sum_{i: \l_{\a i} \in I_\g} \fr{|c_{\a i}|^2}{|N_\g|^2} \(\l_{\a i} + \log |N_\g|^2 \) \\
&= \sum_\a \sum_{i: \l_{\a i} \in I_\g} \fr{|c_{\a i}|^2}{|N_\g|^2} \(\l_\g + \log |N_\g|^2 +O(\e_\g)\) \\
&= \l_\g + \log |N_\g|^2 +O(\e_\g), \la{airgap}
\ea
where in the last step we used \er{ngdef}.

{\allowdisplaybreaks
Using \er{eq:1ru3} and \eqref{lidef}, we then simply compute \ba
& \sum_\a e^{(1-n) A_\UV^\a} \tr_u \(p_\a\r_u^\a\)^n \\
={}& \sum_\a e^{(1-n) A_\UV^\a} \sum_\g \sum_{i: \l_{\a i} \in I_\g} \(p_{\g} \fr{|c_{\a i}|^2}{|N_\g|^2}\)^n \\
={}& \sum_\g p_{\g}^n \sum_\a \sum_{i: \l_{\a i} \in I_\g} \fr{|c_{\a i}|^2}{|N_\g|^2} \(\fr{e^{-\l_{\a i}}}{|N_\g|^2}\)^{n-1} \\
={}& \sum_\g p_{\g}^n \sum_{\a} \sum_{i: \l_{\a i} \in I_\g} \fr{|c_{\a i}|^2}{|N_\g|^2} e^{(1-n)A_\IR^\g} e^{(n-1)\(A_\IR^\g -\l_{\a i} -\log |N_\g|^2\)} \\
={}& \sum_\g p_{\g}^n e^{(1-n)A_\IR^\g} \sum_{\a} \sum_{i: \l_{\a i} \in I_\g} \fr{|c_{\a i}|^2}{|N_\g|^2}  \[1+ (n-1)\(A_\IR^\g -\l_{\a i} -\log |N_\g|^2\) +O\(|n-1|^2 \e_\g^2\)\] \la{taylor}\\
={}& \sum_\g p_{\g}^n e^{(1-n)A_\IR^\g} \sum_{\a} \sum_{i: \l_{\a i} \in I_\g} \fr{|c_{\a i}|^2}{|N_\g|^2} \[1+ O\(|n-1|^2 \e_\g^2\)\] \\
={}& \sum_\m p_\m^n e^{(1-n)\[A_\IR^\m+ S_n(\r_r^\m)\]}\[1+ O\(|n-1|^2 \e_\m^2\)\],
\ea
where on the fifth line we kept the linear term exact and for higher-order terms we used $A_\IR^\g -\l_{\a i} -\log |N_\g|^2 = O(\e_\g)$ due to $\l_{\a i} = \l_\g +O(\e_\g)$ and \er{airgap}, on the sixth line we used \er{airgexact} and \er{ngdef} to remove the linear term, and on the final line we used \er{ngdef}, the identification $\m\eq\g$, and $S_n(\r_r^\m)=0$. Here we have kept $|n-1|^2$ in the error bound because at the moment we allow $n$ to be an arbitrary complex number; in particular, the relative error above is small as long as $|n-1|^2 \e_\m^2$ is small, even when $n$ itself is large.
}

This establishes \er{Renyieq}. The IR R\'enyi FLM formula \er{rflmir} then follows immediately from the UV version \er{rflmuv} up to relative errors of order $|n-1|^2 \e_\m^2$ in each $\m$ term.  So long as all window sizes $\e_\m$ are chosen to be upper bounded by some $\e$ (and $n$ is real), we find that \eqref{rflmir} holds up to an $O\(|n-1|^2 \e^2\)$ relative error. In such cases, upon taking the logarithm of \er{rflmir} we find the usual R\'enyi entropy $S_n(\td\r_B)$ up to an additive error of size $\fr{1}{|n-1|} \, O\(|n-1|^2 \e^2\)$. In particular, the error size vanishes in the $n\to1$ limit (for any $\e$), reproducing the exact IR FLM formula \er{flmir} that we derived in the previous subsection.

%-------------------------------------------------
\section{A general construction}
\label{sec:general}

In the previous section, we considered the simple example of starting with a single state in the IR. We might call it the `seed' state which we used to construct $\cH_\IR$.

Now we provide a general construction of $\cH_\IR$ from an arbitrary set of seed states $|\y_k\>$ in $\cH_\UV$. 
As in section \ref{sec:onestate}, we take the FLM and R\'enyi FLM formulas to hold exactly in the UV.  Cases where these properties hold only approximately in the UV will be addressed in section \ref{sec:disc}.

The construction has two steps, which we now outline briefly:

\paragraph{Step 1:} Construct a $\cH_\IR^{\pre} \subset \cH_\UV$ that contains the seed states, has a decomposition
\be\la{hirpre}
\cH_\IR^\pre =\bigoplus_\b \cH_r^\b \otimes \cH_\rb^\b,
\ee
and satisfies ``complementary recovery'' with respect to $\cH_\UV$ in the sense that the algebra $\cA_r^\pre := \bigoplus_\b \cB(\cH_r^\b)$ can be recovered exactly from $\cA_u$ (meaning that every operator $\cO_r \in \cA_r^\pre$ is represented by a corresponding $\cO_u \in \cA_u$, i.e., $\cO_r |\y\>= \cO_u |\y\>$ and $\cO_r^\dag |\y\>= \cO_u^\dag |\y\>$ for all $|\y\>\in \cH_\IR^{\pre}$), and the commutant $\cA_\rb^\pre := \bigoplus_\b \cB(\cH_\rb^\b)$ can be recovered from $\cA_\ub$. Note that this complementary recovery is a property of the embedding of $\cH_\IR^{\pre}$ into $\cH_\UV$ and does not involve the boundary. In this section we will always use the term `complementary recovery' in this sense (or a similar sense with $\cH_\IR^{\pre}$ replaced by $\cH_\IR$ once we have constructed the latter), unless we state otherwise.

\paragraph{Step 2:} Construct $\cH_\IR$ as an extension of $\cH_\IR^{\pre}$ by including appropriate ``fixed-area'' truncations of states therein.

\bigskip
The above two steps will be described in detail in subsections \re{sec:genpre} and \re{sec:genchop}, respectively. We will then derive the FLM and R\'enyi FLM formulas in the IR in the remaining two subsections.

\subsection{Construction of $\cH_\IR^{\pre}$}
\la{sec:genpre}

This step was trivial in the simple example of a single seed state studied in the previous section, where $\cH_\IR^{\pre}$ is simply the one-dimensional Hilbert space containing the seed state.  In general, we are interested in starting with an arbitrary number of seed states $|\y_k\>$. Let $\cH^\seed$ be their linear span.

It might be tempting to define our desired $\cH_\IR^{\pre}$ to be simply $\cH^\seed$. However, we cannot generally do so because the embedding of an arbitrary subspace $\cH^\seed$ into $\cH_\UV$ does not necessarily satisfy complementary recovery. As a simple example where complementary recovery fails, let us take $\cH_u = \cH_\ub = \bC^2$ (one qubit) so that $\cH_\UV=\bC^2 \otimes \bC^2$ (two qubits).  We also take $\cH^\seed = \operatorname{span}\{|00\>, \fr{1}{\sqrt{2}}(|01\>+|10\>)\}$. It is then straightforward to confirm that there is no decomposition $\cH^\seed =\bigoplus_\b \cH_r^\b \otimes \cH_\rb^\b$ that satisfies complementary recovery.

Instead, we will construct $\cH_\IR^{\pre}$ as a suitable extension of $\cH^\seed$ by adding more states (in $\cH_\UV$). In fact, we will find the smallest such extension.

To do so, let us first understand the consequences of complementary recovery and, in particular, which states are required to be added into $\cH_\IR^{\pre}$. Using Theorem 1 of \cite{Kamal:2019skn} (which generalizes Theorem 5.1 of \cite{Harlow:2016vwg}),\footnote{We will henceforth assume $\cH_\UV$ is finite-dimensional as in \cite{Kamal:2019skn,Harlow:2016vwg}.} we find that complementary recovery of $\cH_\IR^{\pre}$ with the decomposition \er{hirpre} ensures that we have decompositions
\be\la{huub}
\cH_u^\a = \(\bigoplus_\b \cH_r^\b \otimes \cH_o^{\a\b}\) \oplus \cH_{u,\rest}^\a,\qu
\cH_\ub^\a = \(\bigoplus_\b \cH_\rb^\b \otimes \cH_\ob^{\a\b}\) \oplus \cH_{\ub,\rest}^\a,
\ee
such that for each $\b$ (taken from an index set chosen to be independent of $\a$), the spaces $\cH_o^{\a\b}$ and $\cH_\ob^{\a\b}$ both have nonzero dimensions for at least one $\a$.  The above theorems also guarantee the existence of  unitary transformations $U \in \cA_u$, $U' \in \cA_\ub$, and (for each $\b$) a state $|\c\>_{o\ob}^\b$ in
\be\la{hoob}
\cH_{o\ob}^\b := \bigoplus_\a \cH_o^{\a\b} \otimes \cH_\ob^{\a\b},
\ee
such that the embedding of $\cH_\IR^{\pre}$ into $\cH_\UV$ is given by
\be\la{embedu}
|ij\>_{r\rb}^\b \mapsto UU' |i\>_r^\b |\c\>_{o\ob}^\b |j\>_\rb^\b,
\ee
where $|i\>_r^\b$, $|j\>_\rb^\b$ denote orthonormal bases of $\cH_r^\b$, $\cH_\rb^\b$, respectively.\footnote{We chose the names $o$, $\ob$ to suggest that they are integrated out under the RG flow. Moreover, in \er{embedu} we suggestively put $|\c\>_{o\ob}^\b$ in the middle (corresponding to UV, or short-distance, degrees of freedom to be integrated out), separating $|i\>_r^\b$ and $|j\>_\rb^\b$ (corresponding to IR, or long-distance, degrees of freedom to be kept).} 

We can always absorb $U$, $U'$ by choosing the factors $\cH_r^\b$, $\cH_\rb^\b$, $\cH_o^{\a\b}$, and $\cH_\ob^{\a\b}$ appropriately in \er{huub}, and henceforth we will assume that we have done so, dropping $U$, $U'$ from \er{embedu} to instead write
\be\la{embed}
|ij\>_{r\rb}^\b \mapsto |i\>_r^\b |\c\>_{o\ob}^\b |j\>_\rb^\b.
\ee
An immediate consequence of \er{embed} is that, for any density operator $\r$ on $\cH_\IR^\pre$, its UV reduced state $\r_u$ is completely determined by its IR reduced state $\r_r = \bigoplus_\b p_\b \r_r^\b$ (with $\tr \r_r^\b=1$) as
\be\la{ruro}
\r_u = \bigoplus_\b p_\b \r_r^\b \otimes \c_o^\b,
\ee
where $\c_o^\b = \tr_\ob |\c\>_{o\ob}^\b \<\c|_{o\ob}^\b$ is the reduced density operator of $|\c\>_{o\ob}^\b$ with respect to the algebra $\cA_o^\b :=\bigoplus_\a \cB(\cH_o^{\a\b})$.

As a result, for any two density operators $\r$, $\check\r$ on $\cH_\IR^\pre$ for which the support\footnote{The support of an operator is defined as the orthogonal complement of its kernel. For a Hermitian operator it is the same as its range (or image).} $\supp(\check\r_u)$ of $\check\r_u$ contains $\supp(\r_u)$, we find
\be\la{rsinv}
\r_u \check\r_u^{-1} = \bigoplus_\b \[p_\b \r_r^\b (\check p_\b \check\r_r^\b)^{-1}\] \otimes P_{\supp(\c_o^\b)} ,
\ee
where $\check p_\b$ is defined by $\check\r_r = \bigoplus_\b \check p_\b \check\r_r^\b$, and the inverse of a Hermitian operator $\cO$ is defined to be the standard inverse on $\supp(\cO)$ and to annihilate its kernel,\footnote{Such an inverse is a bounded operator because $\cH_\UV$ is finite-dimensional.} so that $\cO^{-1} \cO = \cO \cO^{-1} = P_{\supp(\cO)}$ with $P_{\supp(\cO)}$ the (orthogonal) projection onto $\supp(\cO)$.
The point is that \er{rsinv} always acts within (i.e., preserves) $\cH_\IR^\pre$. This provides a simple way of verifying that in our previous two-qubit example, $\cH^\seed$ does not satisfy complementary recovery: taking $\r$, $\check\r$ to be the density operators for $\fr{1}{\sqrt{3}}(|00\>+|01\>+|10\>)$, $\fr{1}{\sqrt{2}}(|01\>+|10\>)$, respectively, we find $\r_u \check\r_u^{-1} = \fr{2}{3} \begin{psmallmatrix} 2 & 1\\ 1 & 1 \end{psmallmatrix}$ which clearly does not preserve $\cH^\seed$.

Our general strategy is then to enlarge $\cH^\seed$ by adding the images of such $\r_u \s_u^{-1}$.  Concretely, we start with $\cH^\seed$ and define $\s$ to be the maximally mixed state on $\cH^\seed$. We then `bootstrap' the following subalgebra of $\cA_u$:
\be\la{aseed}
\cA_u^\seed := \vN \lt\{\r_u \s_u^{-1}: \forall \text{ density operator $\r$ on $\cH^\seed$}\rt\},
\ee
where $\vN(S)$ denotes the von Neumann algebra generated\footnote{Note that a von Neumann algebra always contains the identity operator (here $\mathbbm 1_u$) and is closed under linear combination, multiplication, and Hermitian conjugation. For \er{aseed}, it is sufficient to use a smaller set of generators $\{\r_u \s_u^{-1}: \r=|\y_k\>\<\y_{k'}|\}$ where $\{|\y_k\>\}$ is a basis of $\cH^\seed$.} by a set $S$, and $\s$ being maximally mixed guarantees $\supp(\r_u) \subset \supp(\s_u)$, because every such $\r$ can be written as a mixture of pure states $|\y_k\>\<\y_k|$, each of which has a reduced state on $u$ that is supported within $\supp(\s_u)$. Now we enlarge $\cH^\seed$ by defining
\be\la{enlarged}
\cH^{\text{enlarged}} := \cA_u^\seed \cdot \cH^\seed,
\ee
consisting of all states obtained by operators in $\cA_u^\seed$ acting on $\cH^\seed$.

In principle, we should repeat the procedure in the previous paragraph by using $\cH^{\text{enlarged}}$ as the new $\cH^{\seed}$, bootstrapping a potentially larger $\cA_u^\seed$ via \er{aseed}, and constructing a potentially larger $\cH^{\text{enlarged}}$ via \er{enlarged}. We will iterate these steps until $\cH^{\text{enlarged}}$ stops growing (as it must terminate within a finite number of iterations, because $\cH_\UV$ is finite-dimensional), and let $\cH_\IR^\pre$ be the final $\cH^{\text{enlarged}}$ and $\cA_u^\pre$ be the final $\cA_u^\seed$.\footnote{Alternatively, we can avoid such iterations by directly defining $\cA_u^\pre := \vN \{\r_u^{(1)}\cd \r_u^{(m)} \s_u^{-m}: \forall \text{ $m\geq 1$ density operators $\r^{(1)}, \cd, \r^{(m)}$ on $\cH^\seed$}\}$ and $\cH_\IR^\pre := \cA_u^\pre \cdot \cH^\seed$. It is straightforward to show that this definition is equivalent to the iterative construction above.}

Therefore, \er{aseed} becomes
\be\la{aupre}
\cA_u^\pre := \vN \lt\{\r_u \s_u^{-1}: \forall \text{ density operator $\r$ on $\cH_\IR^\pre$}\rt\},
\ee
where $\s$ is now the maximally mixed state on $\cH_\IR^\pre$, and \er{enlarged} becomes $\cH_\IR^\pre = \cA_u^\pre \cdot \cH_\IR^\pre$ meaning that $\cA_u^\pre$ acts within $\cH_\IR^\pre$. Note that operators in $\cA_u^\pre$ are defined to act on $\cH_\UV$. Therefore we define
\be\la{arpre}
\cA_r^\pre := \cA_u^\pre \big|_{\cH_\IR^\pre},
\ee
consisting of the restriction of every operator in $\cA_u^\pre$ to $\cH_\IR^\pre$. Thus, $\cA_r^\pre$ is a subalgebra of $\cB(\cH_\IR^\pre)$ and we use it to define the decomposition \er{hirpre} with $\cA_r^\pre = \bigoplus_\b \cB(\cH_r^\b)$, $\cA_\rb^\pre = \bigoplus_\b \cB(\cH_\rb^\b)$.

A useful property to be used momentarily is that, for any $\cO_u \in \cA_u^\pre$, we have
\be\la{faithful}
\cO_u \big|_{\cH_\IR^\pre}=0 \qu \Rightarrow \qu \cO_u P_{\supp(\s_u)} =0.
\ee
In other words, if $\cO_u$ annihilates $\cH_\IR^\pre$, it must annihilate $\supp(\s_u)$. To show this, note that our maximally-mixed state $\s$ is a state on $\cH_\IR^\pre$, so such an $\cO_u$ must satisfy $\cO_u \s=0$. But then $\tr_\ub \cO_u \s = \cO_u \s_u$ must also vanish. Therefore $\cO_u \s_u \s_u^{-1} = \cO_u P_{\supp(\s_u)}$ vanishes, where we used $\s_u^{-1} \s_u =P_{\supp(\s_u)}$. This shows \er{faithful}.

The following theorem guarantees that the construction above satisfies complementary recovery. Moreover, the construction gives the smallest such $\cH_\IR^\pre$ that contains the seed states because we only added states that are required to be included by complementary recovery.

\begin{theorem}
The construction above satisfies complementary recovery: $\cA_r^\pre$ can be recovered from $\cA_u$ and $\cA_\rb^\pre$ can be recovered from $\cA_\ub$.
\end{theorem}

\begin{proof}
The recoverability of $\cA_r^\pre$ from $\cA_u$ follows immediately from the definition \er{arpre} and $\cA_u^\pre \subset \cA_u$, as \er{arpre} ensures that every $\cO_r \in \cA_r^\pre$ can be written as $\cO_u |_{\cH_\IR^\pre}$ with some $\cO_u \in \cA_u^\pre \subset \cA_u$, and is therefore represented by that $\cO_u$.

To prove the recoverability of $\cA_\rb^\pre$ from $\cA_\ub$, we first note that we need only show that our setup satisfies statement (iii) in Theorem 1 of \cite{Kamal:2019skn}, which says that for any operator $X_u \in \cA_u$ we have $P_{\cH_\IR^\pre} X_u P_{\cH_\IR^\pre} \in \cA_r^\pre$, with $P_{\cH_\IR^\pre}$ the projection onto $\cH_\IR^\pre$. To prove this statement, it is then sufficient to verify that $[P_{\cH_\IR^\pre} X_u P_{\cH_\IR^\pre}, \cO_\rb] =0$ for any $\cO_\rb \in \cA_\rb^\pre$. Moreover, it is enough to establish this for any Hermitian $\cO_\rb \in \cA_\rb^\pre$, as any operator is a (complex) linear combination of two Hermitian operators.

We will proceed by showing that any unitary $U_\rb \in \cA_\rb^\pre$ acting on any state $|\y\> \in \cH_\IR^\pre$ must preserve the reduced state on $u$:
\be\la{preserve}
\tr_\ub \(U_\rb |\y\>\<\y| U_\rb^\dag\) = \tr_\ub |\y\>\<\y|.
\ee
To see that \er{preserve} is sufficient, take $U_\rb = e^{i \l \cO_\rb}$ with $\l$ an arbitrary real number. Taking the expectation value of any $X_u \in \cA_u$ using both sides of \er{preserve}, we find
\be
\<\y| e^{-i\l \cO_\rb} X_u e^{i\l \cO_\rb} |\y\> = \<\y| X_u |\y\>,
\ee
which, upon expanding to linear order in $\l$, becomes
\be
\<\y| [X_u, \cO_\rb] |\y\> =0.
\ee
Since $|\y\>$ is an arbitrary state in $\cH_\IR^\pre$, this ensures $[P_{\cH_\IR^\pre} X_u P_{\cH_\IR^\pre}, \cO_\rb] =0$.

It only remains to prove \er{preserve}. As a first step, we claim that every density operator $\r$ on $\cH_\IR^\pre$ has a reduced state on $u$ of the form
\be\la{ru2}
\r_u = \cO_u \s_u, \qu
\cO_u \in \cA_u^\pre.
\ee
To see that this is so, recall that \er{aupre} requires $\r_u \s_u^{-1}$ to be some $\cO_u \in \cA_u^\pre$. Multiplying both by $\s_u$ on the right and using $\s_u^{-1} \s_u =P_{\supp(\s_u)}$, we find
\be
\r_u P_{\supp(\s_u)} = \cO_u \s_u,
\ee
which is equivalent to \er{ru2} because $\supp(\r_u) \subset \supp(\s_u)$.

Next, we show further that \er{ru2} requires $\r_r$ to be
\be\la{rror}
\r_r = \cO_r \s_r, \qu
\cO_r := \cO_u \big|_{\cH_\IR^\pre}.
\ee
This can be established by recalling that every $X_r \in \cA_r^\pre$ is represented by some $X_u \in \cA_u^\pre \subset \cA_u$. Using \er{ru2}, we then find
\bm
\tr_{r} X_r \r_r = \tr_{r\rb} X_r \r = \tr_{u\ub} X_u \r = \tr_{u} X_u \r_u = \tr_{u} X_u \cO_u \s_u = \tr_{u\ub} X_u \cO_u \s \\= \tr_{r\rb} X_r \cO_r \s = \tr_{r} X_r \cO_r \s_r,
\em
where in passing to the second line we used that $X_r$ and $\cO_r$ are represented by $X_u$, $\cO_u$, respectively, and that $\s$ is a state on $\cH_\IR^\pre$. Since $X_r$ is an arbitrary operator in $\cA_r^\pre$, we obtain \er{rror} as desired.

Let us now apply \er{ru2}, \er{rror} to $\r=|\y\> \<\y|$, $\td\r= U_\rb |\y\>\<\y| U_\rb^\dag$ in order to write
\ba
&\r_u = \cO_u \s_u, \qu
\td\r_u = \td\cO_u \s_u, \qu
\cO_u, \td\cO_u \in \cA_u^\pre,\\
&\r_r = \cO_r \s_r, \qu
\td\r_r = \td\cO_r \s_r, \qu
\cO_r := \cO_u \big|_{\cH_\IR^\pre}, \qu
\td\cO_r := \td\cO_u \big|_{\cH_\IR^\pre}.
\ea
Since $U_\rb\in \cA_\rb^\pre$ is unitary, it cannot change the reduced state on $r$:
\be
\r_r = \td\r_r \qu\Rightarrow\qu
\cO_r \s_r = \td\cO_r \s_r \qu\Rightarrow\qu
\cO_r = \td\cO_r,
\ee
where in the second step we used $\s_r$ is invertible (as $\s$ is maximally mixed on $\cH_\IR^\pre$). Using \er{faithful}, we find
\be
\cO_r = \td\cO_r \qu\Rightarrow\qu
(\cO_u - \td\cO_u) P_{\supp(\s_u)}=0 \qu\Rightarrow\qu
\cO_u \s_u = \td\cO_u \s_u \qu\Rightarrow\qu
\r_u = \td\r_u
\ee
which establishes the advertised result \er{preserve}.
\end{proof}

An intuitive way of understanding the above complementary recovery property is as follows. The recoverability of $\cA_r^\pre$ from $\cA_u$ is automatic: $\cA_r^\pre$ is defined to contain only information that can be recovered from $\cA_u$ and in this sense is not too large. The recoverability of $\cA_\rb^\pre$ from $\cA_\ub$ is a statement that $\cA_r^\pre$ is not too small -- \er{aupre} ensures that it contains $\r_u \s_u^{-1}$ (some measure of differences in reduced states) for all states $\r$ in $\cH_\IR^\pre$ and thus has all the information that can be recovered from $\cA_u$ for those states.

As discussed earlier in this subsection, from complementary recovery we are guaranteed that we can decompose $\cH_u^\a$, $\cH_\ub^\a$ as in \er{huub}, and that there exists a state $|\c\>_{o\ob} \in \cH_{o\ob}^\b$ for each $\b$ such that the embedding of $\cH_\IR^{\pre}$ into $\cH_\UV$ is given by \er{embed}.

According to Theorem 2 of \cite{Kamal:2019skn} (which generalizes Theorem 1.1 of \cite{Harlow:2016vwg}), the complementary recovery property shown above immediately leads to (and is in fact equivalent to) an FLM formula\footnote{This is an FLM formula for the bulk entropy in the UV, not the boundary entropy.} with some ``area operator'' (which we might call $\D A^\pre$) in the center of $\cA_r^\pre$:
\be\la{flmur}
S(\r_u) = \<\D A^\pre\>_\r + S(\r_r^\pre)
\ee
as well as a version with $u$, $r$ replaced by $\ub$, $\rb$, respectively. Here $\r$ is an arbitrary state on $\cH_\IR^\pre$ and $\r_r^\pre$ is its reduced density operator with respect to $\cA_r^\pre$.

For our purpose of deriving the IR FLM formula \er{flmir}, we will need to show a slight variant of \er{flmur}:
\be\la{sgeneqpre}
\<A_\UV\>_\r + S(\r_u) = \<A_\IR^\pre\>_\r + S(\r_r^\pre),
\ee
with $A_\IR^\pre$ some area operator in the center of $\cA_r^\pre$.
It is tempting to define $A_\IR^\pre$ as $A_\UV + \D A^\pre$, but the problem is that $A_\UV$ is not generally an operator in $\cA_r^\pre$. This will be remedied shortly by Theorem~\re{flmthm}, which shows that \er{sgeneqpre} holds with the following $A_\IR^\pre$:
\be
A_\IR^\pre := \bigoplus_\b A_\IR^{\pre,\b} \mathbbm 1^\b,\qu
A_\IR^{\pre,\b} = S(\c_o^\b) + \prescript{\b}{o\ob}{\<}\c|A_\UV|\c\>_{o\ob}^\b.
\ee

Although \er{sgeneqpre} leads to our desired IR FLM formula \er{flmir}, its R\'enyi version \er{rflmfixed} does not generally hold in $\cH_\IR^\pre$ since this space need not have a basis of states characterized by approximately flat boundary entanglement spectra. In the following subsection, we will therefore extend $\cH_\IR^\pre$ to a larger Hilbert space $\cH_\IR$ that has such a basis and which satisfies the R\'enyi FLM formula in the IR.

Before proceeding, we emphasize that the entire purpose of this subsection was to find an appropriate extension $\cH_\IR^\pre$ of the span of the seed states such that $\cH_\IR^\pre$ satisfies complementary recovery and whose embedding in $\cH_\UV$ thus takes the form \er{embed}. In doing so, we made technical assumptions such as the finite dimensionality of $\cH_\UV$. However, even in cases where such assumptions do not hold, one may bypass the procedure in this subsection and proceed to the next subsection so long as one can directly identify some $\cH_\IR^\pre$ whose embedding in $\cH_\UV$ takes the form \er{embed}.

\subsection{Construction of $\cH_\IR$}
\la{sec:genchop}

We now construct the desired extension $\cH_\IR$ of $\cH_\IR^\pre$, generalizing the procedure used in Section~\re{sec:onestate}.

Let us first note that, for each $\b$, we have the decomposition \er{hoob},
\be\la{hoob2}
\cH_{o\ob}^\b = \bigoplus_\a \cH_o^{\a\b} \otimes \cH_\ob^{\a\b},
\ee
and a state $|\c\>_{o\ob}^\b \in \cH_{o\ob}^\b$ with a Schmidt decomposition of the form
\be\la{cschmb}
|\c\>_{o\ob}^\b = \bigoplus_\a \sum_i c_{\a i}^\b |i\>_o^{\a\b} |i\>_\ob^{\a\b},
\ee
where $c_{\a i}^\b\neq 0$ by convention (i.e., we only keep nonzero terms in the sum), and where $\{|i\>_o^{\a\b}\}$, $\{|i\>_\ob^{\a\b}\}$ are orthonormal sets in $\cH_o^{\a\b}$, $\cH_\ob^{\a\b}$, respectively. The reduced density operator $\c_o^\b$ in $\cA_o^\b$ of the above state is
\be
\c_o^\b = \bigoplus_\a \sum_i |c_{\a i}^\b|^2 \, |i\>_o^{\a\b} \<i|_o^{\a\b},
\ee
and we define the corresponding modular Hamiltonian as $K_{\c_o^\b} := -\log \c_o^\b$.

Recall that $A_\UV$ is a central operator of the form \er{auv}:
\be\la{auv2}
A_\UV = \bigoplus_\a A_\UV^\a \mathbbm 1^\a,
\ee
where $\mathbbm 1^\a$ is the identity operator on the $\a$-sector, defined previously as the subspace $\cH_u^\a \otimes \cH_\ub^\a$ of $\cH_\UV$. However, $\mathbbm 1^\a$ can also be viewed as the identity on the subspace $\cH_o^{\a\b} \otimes \cH_\ob^{\a\b}$ of $\cH_{o\ob}^\b$, and hence $A_\UV$ can also be viewed as an operator on $\cH_{o\ob}^\b$. We will use this convention below.

As a result, on the support of $\c_o^\b$ (viewed as a subspace of $\cH_{o\ob}^\b$), we have
\be
A_\UV + K_{\c_o^\b} = \bigoplus_\a \sum_i \l_{\a i}^\b |i\>_o^{\a\b} \<i|_o^{\a\b},
\ee
where, much as in section \ref{sec:onestate}, we have introduced
\be\la{libdef}
\l_{\a i}^\b := A_\UV^\a -\log |c_{\a i}^\b|^2.
\ee
A helpful way of understanding these $\l^\b_{\a i}$ is that they are the eigenvalues of the `code modular Hamiltonian' (as defined in footnote~\re{codeK}) for the code that maps $\cH_\IR^\pre$ to the boundary.

For each $\b$, we now chop the states in 
$\cH_\IR^\pre$ into pieces by grouping the eigenvalues $\l_{\a i}^\b$ into `bins' of some small width. In other words, we divide the entire range of $\l_{\a i}^\b$ into intervals labeled by some index $\g$. Each bin corresponds to an interval $I_\g^\b = [\l_\g^\b-\e_\g^\b/2, \l_\g^\b+\e_\g^\b/2)$ for the eigenvalues, and defines a truncated state by only keeping those terms in \er{cschmb} with $\l_{\a i}^\b$ in the interval $I_\g^\b$:
\be\la{cgbdef}
|\c^\g\>_{o\ob}^\b := \fr{1}{N_\g^\b} \bigoplus_\a \sum_{i: \l_{\a i}^\b \in I_\g^\b} c_{\a i}^\b |i\>_o^{\a\b} |i\>_\ob^{\a\b},
\ee
where $N_\g^\b$ is a normalization constant satisfying
\be\la{ngbdef}
|N_\g^\b|^2 = \sum_\a \sum_{i: \l_{\a i}^\b \in I_\g^\b} |c_{\a i}^\b|^2.
\ee
We then define the IR Hilbert space as the linear span of the pieces defined by the above chopping procedure:
\be\la{hirdef}
\cH_\IR = \operatorname{span} \{|i\>_r^\b |\c^\g\>_{o\ob}^\b |j\>_\rb^\b: \g \in J^\b,\, |i\>_r^\b \in \cH_r^\b,\, |j\>_\rb^\b \in \cH_\rb^\b\},
\ee
where $J^\b$ is the set of all indices $\g$ for a given $\b$. In the rest of the paper, sums over $\g$ are implicitly over all elements of $J^\b$ (where $\b$ would be clear from the context). Comparing \er{hirdef} with \er{embed}, we see that $\cH_\IR$ is indeed an extension of $\cH_\IR^\pre$.

Our next step is to verify that this $\cH_\IR$ achieves the goals described in the Introduction. First, by construction it can be decomposed as
\be
\cH_\IR = \bigoplus_{\b,\g} \cH_r^{\b\g} \otimes \cH_\rb^{\b\g},
\ee
where $\cH_r^{\b\g}$ ($\cH_\rb^{\b\g}$) is isomorphic to $\cH_r^\b$ ($\cH_\rb^\b$) for all values of $\g$. This decomposition obviously reproduces our desired \er{hir} saying $\cH_\IR = \bigoplus_\m \cH_r^\m \otimes \cH_\rb^\m$: we simply identify $\m$ with the pair $(\b,\g)$,
\be
\m \eq (\b,\g).
\ee
The IR algebras $\cA_r = \bigoplus_\m \cB(\cH_r^\m)$, $\cA_\rb = \bigoplus_\m \cB(\cH_\rb^\m)$ are commutants of each other, and their intersection is the center $\cZ_r$ generated by $\mathbbm 1^\m$.

We now show that the embedding of $\cH_\IR$ into $\cH_\UV$ satisfies complementary recovery. A simple way to see this is to verify statement (i) in Theorem 1 of \cite{Kamal:2019skn} for both $u$ and $\bar u$. It is sufficient to check that we can decompose
\be\la{huubm}
\cH_u^\a = \(\bigoplus_\m \cH_r^\m \otimes \cH_o^{\a\m}\) \oplus \cH_{u,\rest}^\a,\qu
\cH_\ub^\a = \(\bigoplus_\m \cH_\rb^\m \otimes \cH_\ob^{\a\m}\) \oplus \cH_{\ub,\rest}^\a,
\ee
such that for each $\m$, the spaces $\cH_o^{\a\m}$ and $\cH_\ob^{\a\m}$ both have nonzero dimensions for at least one $\a$, and to show that (for each $\m$) there is a state $|\c\>_{o\ob}^\m$ in
\be\la{hoobm}
\cH_{o\ob}^\m := \bigoplus_\a \cH_o^{\a\m} \otimes \cH_\ob^{\a\m},
\ee
such that the embedding of $\cH_\IR$ into $\cH_\UV$ is given by
\be\la{embedm}
|ij\>_{r\rb}^\m \mapsto |i\>_r^\m |\c\>_{o\ob}^\m |j\>_\rb^\m,
\ee
where $|i\>_r^\m$, $|j\>_\rb^\m$ denote orthonormal bases of $\cH_r^\m$, $\cH_\rb^\m$, respectively. It is straightforward to see that this statement follows from \er{huub}--\er{embed}, together with the way we refined $\b$-sectors into $\m$-sectors. In particular, \er{huubm} follows from \er{huub} with $\cH_o^{\a\m}\eq \cH_o^{\a\b\g}$ defined as the intersection of $\cH_o^{\a\b}$ and the support of the reduced state $\c_o^\m \eq \c_o^{\b\g} = \tr_\ob |\c^\g\>_{o\ob}^\b \<\c^\g|_{o\ob}^\b$ (and similarly for $\cH_\ob^{\a\m}$), the state $|\c\>_{o\ob}^\m$ is simply $|\c^\g\>_{o\ob}^\b$, and the basis $|i\>_r^\m$ is simply the basis $|i\>_r^\b$ of $\cH_r^\b$ because $\cH_r^{\m}$ is isomorphic to $\cH_r^\b$ for all values of $\g$ (and similarly for $|j\>_\rb^\m$).

As before, an immediate consequence of \er{embedm} is that for any density operator $\r$ on $\cH_\IR$, its UV reduced state $\r_u$ is completely determined by its IR reduced state
\be\la{rmdecom}
\r_r = \bigoplus_\m p_\m \r_r^\m
\ee
as
\be\la{rurm}
\r_u = \bigoplus_\m p_\m \r_r^\m \otimes \c_o^\m,
\ee
where $\tr \r_r^\m=1$, and again $\c_o^\m$ is the reduced density operator of $|\c\>_{o\ob}^\m$ with respect to the algebra $\cA_o^\m :=\bigoplus_\a \cB(\cH_o^{\a\m})$. In particular, this implies that a certain notion of `subregion orthogonality' is preserved by the IR-to-UV map: two states $\r^{(1)}$, $\r^{(2)}$ in $\cH_\IR$ are orthogonal on $r$ (i.e., $\r^{(1)}_r \r^{(2)}_r=0$) if and only if they are orthogonal on $u$ (i.e., $\r^{(1)}_u \r^{(2)}_u=0$); a similar statement holds for the complementary regions. Therefore, our construction ensures that, as long as the UV-to-boundary map preserves subregion orthogonality (i.e., orthogonality on $u$ is equivalent to orthogonality on the boundary region $B$, and similarly for the complementary regions), the IR code will inherit this property (i.e., orthogonality on $r$ will be equivalent to orthogonality on $B$, and similarly for the complementary regions). A special case of this property is that IR states from different $\m$-sectors -- which are automatically orthogonal on $r$ and on $\rb$ -- are mapped to boundary states that are orthogonal on boundary subsystems $B$ and $\Bb$; this will be called `boundary subsystem orthogonality' in \cite{JLMS} and will prove useful therein.

Complementary recovery of the IR-to-UV map follows manifestly from \er{embedm}: $\cA_r$ ($\cA_\rb$) can be recovered from $\cA_u$ ($\cA_\ub$) as the restriction of an appropriate subalgebra to $\cH_\IR$.

\subsection{FLM formula}
\la{sec:genflm}

According to Theorem 2 of \cite{Kamal:2019skn}, complementary recovery leads to (and is in fact equivalent to) an FLM formula with some ``area operator'' (which we might call $\D A$) in the center of $\cA_r$:
\be\la{flmur2}
S(\r_u) = \<\D A\>_\r + S(\r_r),\qqu
S(\r_\ub) = \<\D A\>_\r + S(\r_\rb),
\ee
where $\r$ is an arbitrary state on $\cH_\IR$ and $\r_r$ is its reduced density operator with respect to $\cA_r$. By comparing \er{rmdecom} with \er{rurm}, we find an explicit expression for $\D A$:
\be\la{dadef}
\D A := \bigoplus_\m \D A^\m \mathbbm 1^\m,\qqu
\D A^\m := S(\c_o^\m).
\ee

For our purpose of deriving the IR FLM formula \er{flmir}, we will need to show a slight variant of \er{flmur2}. We will do so by proving the following theorem.

\begin{theorem}\la{flmthm}
Let $\cH_\UV$ be a finite-dimensional Hilbert space and $\cH_\IR$ be a subspace. Let $\cA_u$, $\cA_r$ be von Neumann algebras on $\cH_\UV$, $\cH_\IR$, respectively, and $A_\UV$ be an operator in the center of $\cA_u$. Then complementary recovery (i.e., $\cA_r$ and its commutant $\cA_\rb$ can be recovered from $\cA_u$ and its commutant $\cA_\ub$, respectively) is equivalent to the existence of an operator $A_\IR$ in the center of $\cA_r$ such that
\be\la{sgeneq2}
\<A_\UV\>_\r + S(\r_u) = \<A_\IR\>_\r + S(\r_r),\qqu
\<A_\UV\>_\r + S(\r_\ub) = \<A_\IR\>_\r + S(\r_\rb)
\ee
for any density operator $\r$ on $\cH_\IR$.
\end{theorem}

\begin{proof}
\textbf{Complementary recovery $\Rightarrow$ \er{sgeneq2}:} To prove this, recall that according to Theorem 2 of \cite{Kamal:2019skn}, complementary recovery leads to \er{flmur2}. To derive \er{sgeneq2} from \er{flmur2}, we need only establish the existence of an operator $\td A_\UV$ in the center of $\cA_r$ such that for any density operator $\r$ on $\cH_\IR$, we have
\be\la{auvtd}
\<A_\UV\>_\r = \<\td A_\UV\>_\r.
\ee
Our desired result will then follow immediately from the definition
\be\la{airdef2}
A_\IR := \td A_\UV + \D A.
\ee
To establish the existence of such an $\td A_\UV$, first recall that, according to Theorem 1 of \cite{Kamal:2019skn}, complementary recovery guarantees that we can decompose $\cH_u^\a$, $\cH_\ub^\a$ as in \er{huubm}, and that there exists a state $|\c\>_{o\ob}^\m \in \cH_{o\ob}^\m$ for each $\m$ such that the embedding of $\cH_\IR$ into $\cH_\UV$ is given by \er{embedm}. Thus, we define
\be\la{auvtddef}
\td A_\UV := \bigoplus_\m \td A_\UV^\m \mathbbm 1^\m,\qqu
\td A_\UV^\m := \prescript{\m}{o\ob}{\<}\c|A_\UV|\c\>_{o\ob}^\m,
\ee
where $\prescript{\m}{o\ob}{\<}\c|A_\UV|\c\>_{o\ob}^\m$ is defined by decomposing $|\c\>_{o\ob}^\m$ into $\a$-sectors according to $\cH_{o\ob}^\m = \bigoplus_\a \cH_o^{\a\m} \otimes \cH_\ob^{\a\m}$ and using the fact that $A_\UV$ is a c-number on each $\a$-sector. This $\td A_\UV$ is manifestly in the center of $\cA_r$. Finally, we verify \er{auvtd} for any $\r$ on $\cH_\IR$:
\ba
& \<\td A_\UV\>_\r = \tr_{r\rb} (\r \td A_\UV) = \tr_r (\r_r \td A_\UV) = \tr_r \[\(\bigoplus_\m p_\m \r_r^\m \) \td A_\UV\] = \sum_\m p_\m \td A_\UV^\m \\
=& \sum_\m p_\m \prescript{\m}{o\ob}{\<}\c|A_\UV|\c\>_{o\ob}^\m = \sum_\m p_\m \tr_o (\c_o^\m A_\UV) = \tr_u \[\(\bigoplus_\m p_\m \r_r^\m \otimes \c_o^\m\) A_\UV\] \\
=& \tr_{u} (\r_u A_\UV) = \tr_{u\ub} (\r A_\UV) = \<A_\UV\>_\r,
\ea
where we used \er{rmdecom} in the first line and \er{rurm} in going to the last line. Thus, we have shown \er{sgeneq2} with an explicit construction \er{airdef2} of $A_\IR$.

\textbf{\er{sgeneq2} $\Rightarrow$ complementary recovery:} We will be brief in proving this direction as it is not essential for the purposes of this paper. We use the same method as in proving Theorem 2 of \cite{Kamal:2019skn}: varying \er{sgeneq2} under an infinitesimal perturbation $\d\r$, we find a variant of the JLMS formula for the modular Hamiltonians $K_{\r_u}:= -\log \r_u$, $K_{\r_r}:= -\log \r_r$:
\be
P_{\cH_\IR} (A_\UV + K_{\r_u}) P_{\cH_\IR} = A_\IR + K_{\r_r},
\ee
as well as a version with $u$, $r$ replaced by $\ub$, $\rb$, respectively. Here $P_{\cH_\IR}$ is the projection onto $\cH_\IR$. From this, we find the JLMS formula for the relative entropies:
\be
S(\r_u|\s_u)= S(\r_r|\s_r),\qqu
S(\r_\ub|\s_\ub)= S(\r_\rb|\s_\rb),
\ee
where $\r$, $\s$ are two arbitrary density operators on $\cH_\IR$. This leads to (and is in fact equivalent to) complementary recovery, according to Theorem 2 of \cite{Kamal:2019skn}.
\end{proof}

The FLM-like formula \er{sgeneq2} immediately implies that our desired IR FLM formula \er{flmir} holds to the extent that the UV version \er{flmuv} holds. In particular, if the UV FLM formula holds within some error bar $\e_{FLM}$, the IR version holds within the same error bar.

In the rest of this subsection, let us derive a few useful results on $A_\IR$. First, by combining \er{dadef}, \er{airdef2}, and \er{auvtddef}, we find a more direct expression for $A_\IR$:
\be
A_\IR := \bigoplus_\m A_\IR^\m \mathbbm 1^\m,\qqu
A_\IR^\m := S(\c_o^\m) + \prescript{\m}{o\ob}{\<}\c|A_\UV|\c\>_{o\ob}^\m.
\ee
More explicitly, we may use $|\c\>_{o\ob}^\m = |\c^\g\>_{o\ob}^\b$ and \er{cgbdef} to write $A_\IR^\m$ as
\be
A_\IR^\m =A_\IR^{\b\g} = \sum_\a \sum_{i: \l_{\a i}^\b \in I_\g^\b} \fr{|c_{\a i}^\b|^2}{|N_\g^\b|^2} \(A_\UV^\a - \log \fr{|c_{\a i}^\b|^2}{|N_\g^\b|^2} \).
\ee
Recalling that the eigenvalues $\l_{\a i}^\b$ defined by \er{libdef} are approximately equal to $\l_\g^\b$ within each interval $I_\g^\b$, with an error no greater than $\e_\g^\b/2$, we find an approximate but simpler expression for $A_\IR^\m$:
\ba
A_\IR^\m &= \sum_\a \sum_{i: \l_{\a i}^\b \in I_\g^\b} \fr{|c_{\a i}^\b|^2}{|N_\g^\b|^2} \(\l_{\a i}^\b + \log |N_\g^\b|^2 \) \la{airmexact}\\
&= \sum_\a \sum_{i: \l_{\a i}^\b \in I_\g^\b} \fr{|c_{\a i}^\b|^2}{|N_\g^\b|^2} \(\l_\g^\b + \log |N_\g^\b|^2 +O(\e_\m)\) \\
&= \l_\g^\b + \log |N_\g^\b|^2 +O(\e_\m) \la{airm} \\
&= \log \(\sum_\a e^{A_\UV^\a} D_\m^\a\) +O(\e_\m), \la{airmdim}
\ea
where $\e_\m:=\e_\g^\b$ is a simpler notation that we will use from now on, and on the last line\footnote{The $O(\e_\m)$ error bound in \er{airmdim} can be sharpened: the logarithm overestimates $A_\IR^\m$ by at most $\e_\m^2/8$, as shown in \er{app:area-bound}.} we used \er{libdef} and \er{ngbdef} to find
\ba
|N_\g^\b|^2 &= \sum_\a \sum_{i: \l_{\a i}^\b \in I_\g^\b} e^{A_\UV^\a - \l_{\a i}^\b} \\
&= \sum_\a \sum_{i: \l_{\a i}^\b \in I_\g^\b} e^{A_\UV^\a - \l_\g^\b + O(\e_\m)} = \(e^{-\l_\g^\b} \sum_\a e^{A_\UV^\a} D_\m^\a\) \[1+ O(\e_\m)\], \la{ngbap}
\ea
with $D_\m^\a$ defined as the number of terms in the sum over $i$ in \er{ngbap}, for any given $\a$ and $\m=(\b,\g)$.

\subsection{R\'enyi FLM formula}

Recall that the R\'enyi FLM formula in the UV takes the general form \er{rflmuv}, which we repeat here for the convenience of the reader:
\be\la{rflmuv2}
e^{(1-n)S_n(\td\r_B)} = \sum_\a p_\a^n e^{(1-n)\[A_\UV^\a+ S_n(\r_u^\a)\]},
\ee
where again $p_\a$, $\r_u^\a$ are defined by $\r_u =\bigoplus_\a p_\a \r_u^\a$ and $\tr \r_u^\a=1$.  As noted earlier, this section assumes \er{rflmuv2} to be exact in the UV code.

We now show that the IR R\'enyi FLM formula \er{rflmir} holds approximately for any IR state $\r$:
\be\la{rflmir2}
e^{(1-n)S_n(\td\r_B)} = \sum_\m p_\m^n e^{(1-n)\[A_\IR^\m+ S_n(\r_r^\m)\]} \[1+ O\(|n-1|^2\e_\m^2\)\].
\ee
Here we used the simpler notation $\e_\m:=\e_\g^\b$, and $p_\m$, $\r_r^\m$ are again defined by $\r_r =\bigoplus_\m p_\m \r_r^\m$ and $\tr \r_r^\m=1$.

To proceed, we wish to show that for any state $\r$ in $\cH_\IR$,
\be
\sum_\a p_\a^n e^{(1-n)\[A_\UV^\a+ S_n(\r_u^\a)\]} = \sum_\m p_\m^n e^{(1-n)\[A_\IR^\m+ S_n(\r_r^\m)\]} \[1+ O\(|n-1|^2\e_\m^2\)\],
\ee
which is equivalent to
\be\la{renyieq2}
\sum_\a e^{(1-n) A_\UV^\a} \tr_u \(p_\a\r_u^\a\)^n = \sum_\m e^{(1-n) A_\IR^\m} \tr_r \(p_\m\r_r^\m\)^n \[1+ O\(|n-1|^2\e_\m^2\)\].
\ee

To establish \er{renyieq2}, we first note that the density operator $\c_o^\m$ as an element of $\cA_o^\m :=\bigoplus_\a \cB(\cH_o^{\a\m})$ can be decomposed as $\c_o^\m = \bigoplus_\a \hat\c_o^{\a\m}$, where $\hat\c_o^{\a\m}$ does not generally have trace $1$. Thus for any state $\r$ in $\cH_\IR$, its UV reduced state \er{rurm} can be written as
\be
\r_u = \bigoplus_{\a,\m} p_\m \r_r^\m \otimes \hat\c_o^{\a\m}.
\ee
Comparing this with $\r_u =\bigoplus_\a p_\a \r_u^\a$, we find
\be
p_\a \r_u^\a = \bigoplus_\m p_\m \r_r^\m \otimes \hat\c_o^{\a\m}.
\ee
As a result, we find an exact equality:
\ba
\sum_\a e^{(1-n) A_\UV^\a} \tr_u \(p_\a\r_u^\a\)^n &= \sum_\a e^{(1-n) A_\UV^\a} \tr_u \(\bigoplus_\m p_\m \r_r^\m \otimes \hat\c_o^{\a\m}\)^n \\
&= \sum_{\a,\m} e^{(1-n) A_\UV^\a} \tr_o \(\hat\c_o^{\a\m}\)^n \tr_r \(p_\m\r_r^\m\)^n. \la{renyifact}
\ea

{\allowdisplaybreaks
To complete the derivation of \er{renyieq2}, we need only show
\be
\sum_{\a} e^{(1-n) A_\UV^\a} \tr_o \(\hat\c_o^{\a\m}\)^n = e^{(1-n) A_\IR^\m} \[1+ O\(|n-1|^2\e_\m^2\)\].
\ee
We proceed by using \er{cgbdef} to find an explicit expression for $\hat\c_o^{\a\m}$:
\be
\hat\c_o^{\a\m} = \hat\c_o^{\a\b\g} = \fr{1}{|N_\g^\b|^2} \sum_{i: \l_{\a i}^\b \in I_\g^\b} |c_{\a i}^\b|^2 \, |i\>_o^{\a\b} \<i|_o^{\a\b}.
\ee
Using this result, we write
\ba
& \sum_{\a} e^{(1-n) A_\UV^\a} \tr_o \(\hat\c_o^{\a\m}\)^n \\
={}& \sum_{\a} \sum_{i: \l_{\a i}^\b \in I_\g^\b} e^{(1-n) A_\UV^\a} \(\fr{|c_{\a i}^\b|^2}{|N_\g^\b|^2}\)^n \\
={}& \sum_{\a} \sum_{i: \l_{\a i}^\b \in I_\g^\b} \fr{|c_{\a i}^\b|^2}{|N_\g^\b|^2} \(\fr{e^{-\l_{\a i}^\b}}{|N_\g^\b|^2}\)^{n-1} \\
={}& \sum_{\a} \sum_{i: \l_{\a i}^\b \in I_\g^\b} \fr{|c_{\a i}^\b|^2}{|N_\g^\b|^2} e^{(1-n)A_\IR^\m} e^{(n-1)\(A_\IR^\m -\l_{\a i}^\b -\log |N_\g^\b|^2\)} \la{renyisectorexact}\\
={}& e^{(1-n)A_\IR^\m} \sum_{\a} \sum_{i: \l_{\a i}^\b \in I_\g^\b} \fr{|c_{\a i}^\b|^2}{|N_\g^\b|^2} \[1+ (n-1)\(A_\IR^\m -\l_{\a i}^\b -\log |N_\g^\b|^2\) +O\(|n-1|^2 \e_\m^2\)\] \la{renyitaylor}\\
={}& e^{(1-n)A_\IR^\m} \sum_{\a} \sum_{i: \l_{\a i}^\b \in I_\g^\b} \fr{|c_{\a i}^\b|^2}{|N_\g^\b|^2} \[1+ O\(|n-1|^2 \e_\m^2\)\] \\
={}& e^{(1-n)A_\IR^\m} \[1+ O\(|n-1|^2 \e_\m^2\)\],
\ea
where on the third line we used \er{libdef}, on the fifth line we kept the linear term exact and for higher-order terms we used $A_\IR^\m -\l_{\a i}^\b -\log |N_\g^\b|^2 = O(\e_\m)$ due to $\l_{\a i}^\b = \l_\g^\b +O(\e_\m)$ and \er{airm}, on the sixth line we used \er{airmexact} and \er{ngbdef} to remove the linear term, and on the final line we used \er{ngbdef} again. Here we have kept $|n-1|^2$ in the error bound because at the moment we allow $n$ to be an arbitrary complex number.
}

We have thus established \er{renyieq2}. Our desired IR R\'enyi FLM formula \er{rflmir2} then holds to the extent that the UV version \er{rflmuv2} holds (i.e., the IR formula inherits any additional errors in the UV formula if it is approximate). The relative errors shown in \er{rflmir2} are of order $|n-1|^2 \e_\m^2$ for each $\m$ term; these errors to the R\'enyi FLM were introduced because the entanglement spectra within the eigenvalue windows are only approximately flat. So long as all window sizes $\e_\m$ are chosen to be upper bounded by some $\e$ (and $n$ is real\footnote{For complex $n$, there can be cancellations among the formally-leading terms for different values of $\mu$ so that the final relative error is larger.  However, such cancellations require fine tuning. The relative error will thus remain $O\(|n-1|^2 \e^2\)$ for generic cases; e.g., in the limit $\epsilon \rightarrow 0$ with the IR state held fixed and with $n$ fixed to a generic complex value.}), we find that \eqref{rflmir2} holds up to an $O\(|n-1|^2 \e^2\)$ relative error. In such cases, upon taking the logarithm of \er{rflmir2} we find the usual R\'enyi entropy $S_n(\td\r_B)$ up to an additive error of size $\fr{1}{|n-1|}\, O\(|n-1|^2 \e^2\)$.
As in section \re{sec:onestate}, the error size vanishes in the $n\to1$ limit (for any $\e$), reproducing the exact IR FLM formula \er{flmir} that we derived in the previous subsection.

In appendix~\ref{app:explicit-bounds}, we give sharper, explicit error bounds for the approximation to the IR area operator and the R\'enyi FLM formula.\footnote{This improvement was suggested by ChatGPT.}

%-------------------------------------------------
\section{Examples}
\label{sec:ex}

Here we apply the general construction of section \re{sec:general} to two examples. They are more nontrivial than the single-seed-state example of section \re{sec:onestate} and in particular will demonstrate cases where Step 1 of the construction is nontrivial and produces an $\cH^\pre_\IR$ that is strictly larger than $\cH^\seed$.

Before proceeding, we prove the following lemma about $\cA^\seed_u$ to simplify some of our analysis later.

\begin{lemma}\la{aseedlemma}
Assume as before that $\cH_\UV$ is finite-dimensional. For any two density operators $\r$, $\check{\r}$ on $\cH^\seed$ with $\supp(\r_u) \subset \supp(\check{\r}_u)$, the algebra $\cA^\seed_u$ defined by \er{aseed} contains $\r_u \check{\r}_u^{-1}$. The same statement holds for $\r=|\y_1\>\<\y_2|$ with any two potentially distinct states $|\y_1\>$, $|\y_2\>$ in $\cH^\seed$ (but where $\check{\r}$ remains Hermitian).
Here again the inverse of a Hermitian operator $\cO$ is defined to be the standard inverse on $\supp(\cO)$ and to annihilate its kernel, so that $\cO^{-1} \cO = \cO \cO^{-1} = P_{\supp(\cO)}$.
\end{lemma}

\begin{proof}
Let $\r$ be either a density operator on $\cH^\seed$, or $|\y_1\>\<\y_2|$ with any two states $|\y_1\>$, $|\y_2\>$ in $\cH^\seed$. In both cases, it satisfies $\r_u \s_u^{-1} \in \cA^\seed_u$ (where $\s$ is the maximally mixed state on $\cH^\seed$). In the latter case, this is because $|\y_1\>\<\y_2|$ can be written as a linear combination of density operators.

Let ${\cA^\seed_u}'$ be the commutant of $\cA^\seed_u$ in $\cB(\cH_\UV)$. By the von Neumann bicommutant theorem, it suffices to prove that $\r_u \check{\r}_u^{-1}$ commutes with every $\cO \in {\cA^\seed_u}'$.
To do so, we note that such an $\cO$ by definition commutes with $\r_u \s_u^{-1}, \check{\r}_u \s_u^{-1} \in \cA^\seed_u$ and find
\ba
\cO \r_u \check{\r}_u^{-1} &= \cO \r_u \s_u^{-1} \s_u \check{\r}_u^{-1}
 = \r_u \s_u^{-1} \cO \s_u \check{\r}_u^{-1} = \r_u P_{\supp(\check{\r}_u)} \s_u^{-1} \cO \s_u \check{\r}_u^{-1} \\
&= \r_u \check{\r}_u^{-1} \check{\r}_u \s_u^{-1} \cO \s_u \check{\r}_u^{-1} = \r_u \check{\r}_u^{-1} \cO \check{\r}_u \s_u^{-1} \s_u \check{\r}_u^{-1} = \r_u \check{\r}_u^{-1} \cO \check{\r}_u \check{\r}_u^{-1} \\
&= \r_u \check{\r}_u^{-1} \cO P_{\supp(\check{\r}_u)} = \r_u \check{\r}_u^{-1} P_{\supp(\check{\r}_u)} \cO = \r_u \check{\r}_u^{-1} \cO,
\ea
where in the last step of the first line we used $\supp(\r_u) \subset \supp(\check{\r}_u)$, and in going to the last line we used $P_{\supp(\check{\r}_u)} \in \cA^\seed_u$ (and thus it commutes with $\cO$); to see this, note that $\supp(\check{\r}_u)$ is the range of $\check{\r}_u$, which is the range of $\check{\r}_u \s_u^{-1}$ (because $\supp(\check{\r}_u) \subset \supp(\s_u)$), which in turn is the range of a Hermitian operator $\check{\r}_u \s_u^{-1} (\check{\r}_u \s_u^{-1})^\dag$ in $\cA^\seed_u$, and $\cA^\seed_u$ must contain the projection onto this range (as it contains all spectral projections).
\end{proof}

\subsection{Two seed states with a common $|\c\>$}

In this example, we choose
\be
\cH_\UV = \cH_u \otimes \cH_\ub,\qu
\cH_u = \bC^2 \otimes \cH_o,\qu
\cH_\ub = \bC^2 \otimes \cH_\ob,
\ee
where $\cH_o$, $\cH_\ob$ are Hilbert spaces. Let $\{|0\>, |1\>\}$ be an orthonormal basis of the two-dimensional Hilbert space $\bC^2$, and $|\c\>_{o\ob}$ be an arbitrary state in $\cH_{o\ob} = \cH_o \otimes \cH_\ob$ whose reduced density operator $\c_o$ has full rank.

We will now choose two seed states in $\cH_\UV$. The first is chosen as
\be
|\y_1\> = |00\> \otimes |\c\>_{o\ob},
\ee
where the first (second) $0$ specifies the state on the $\bC^2$ factor in $\cH_u$ ($\cH_\ub$).

If we were to choose the second seed state as $|\y_2\> = |10\> \otimes |\c\>_{o\ob}$, our construction would simply produce $\cH^\pre_\IR = \cH^\pre_r\otimes \cH^\pre_\rb$ with a two-dimensional $\cH^\pre_r$ and a one-dimensional $\cH^\pre_\rb$. If instead we were to choose the second seed state as $|\y_2\> = |11\> \otimes |\c\>_{o\ob}$, our construction would produce $\cH^\pre_\IR = \bigoplus_{\b\in \{1,2\}} \cH_r^\b \otimes \cH_\rb^\b$ consisting of 2 one-dimensional $\b$-sectors. In both of these cases, $\cH^\pre_\IR = \cH^\seed$.

To make the example more nontrivial, we choose the second seed state to be
\be
|\y_2\> = \fr{1}{\sqrt{2}}(|01\>+|10\>) \otimes |\c\>_{o\ob},
\ee

We now build $\cH^\pre_\IR$ using Step 1 of our construction (described in section \re{sec:genpre}). First, we show
\be\la{aseedex1}
\cA^\seed_u = \cB(\bC^2) \otimes (\bC\,\mathbbm 1_o).
\ee
To see this, we use Lemma~\re{aseedlemma} with $\check{\r} = |\y_2\>\<\y_2|$ and thus $\check{\r}_u= \fr{1}{2} \begin{psmallmatrix} 1 & 0\\ 0 & 1 \end{psmallmatrix} \otimes \c_o$. With the choice $\r=|\y_1\>\<\y_1|$, the lemma tells us that $\cA^\seed_u$ contains
\be\la{rcr1}
\r_u \check{\r}_u^{-1} = 2 \begin{pmatrix} 1 & 0\\ 0 & 0 \end{pmatrix} \otimes \mathbbm 1_o.
\ee
With another choice $\r=|\y_1\>\<\y_2|$, the lemma tells us that $\cA^\seed_u$ contains
\be\la{rcr2}
\r_u \check{\r}_u^{-1} = \sqrt 2 \begin{pmatrix} 0 & 1\\ 0 & 0 \end{pmatrix} \otimes \mathbbm 1_o.
\ee
The operators \er{rcr1} and \er{rcr2} already generate the right-hand side of \er{aseedex1}, and $\cA^\seed_u$ cannot be larger than that, so we have shown \er{aseedex1}.

Acting with this $\cA^\seed_u$ on $\cH^\seed$, we find that our construction gives
\be
\cH^\pre_\IR = \operatorname{span} \{|i\>_r \otimes |j\>_\rb \otimes |\c\>_{o\ob}: i, j=0, 1 \},
\ee
which is four-dimensional. Our construction recognized that the original $\cH^\seed$ could not satisfy complementary recovery (as discussed in section \re{sec:genpre}) and added in two additional states so that the resulting $\cH^\pre_\IR$ satisfies complementary recovery.

We then build $\cH_\IR$ by chopping $|\c\>_{o\ob}$ into pieces defined by appropriate eigenvalue windows. This part is similar to what we already did for the single-seed-state example of section \re{sec:onestate}, so we will not repeat it here.

\subsection{Two seed states without a common $|\c\>$}

We now choose two seed states that do not precisely contain any common $|\c\>$ factor.  This is of course the generic case. We will present a large class of examples of this sort that require $\cH^\pre_\IR = \cH_\UV$ (and thus also $\cH_\IR = \cH_\UV$). This includes cases where the two seed states are arbitrarily close to each other and where $\cH_u$, $\cH_\ub$ have arbitrarily large dimensions.

Before proceeding, we prove another useful lemma.

\begin{lemma}\la{alglemma}
Let $D_1$, $D_2$ be two $N$-dimensional diagonal matrices with real, non-degenerate, but otherwise arbitrary diagonal elements. Let $U$ be any $N$-dimensional unitary matrix whose matrix elements are all nonzero. Let $\cO=D_1 +i U D_2 U^\dag$. Then
\be
\vN \{\cO\} = \cB(\bC^N).
\ee
\end{lemma}

\begin{proof}
The von Neumann algebra $\vN \{\cO\}$ generated by $\cO$ obviously contains $D_1$ and $U D_2 U^\dag$. Since these are both Hermitian with non-degenerate eigenvalues, $\vN \{\cO\}$ must contain the projection onto each of their eigenvectors (as it contains all spectral projections). Let $\{|j\>\}$ be the orthonormal basis in which $D_1$ is diagonal. Then $\{|\td k\> = U|k\>\}$ is an orthonormal basis that diagonalizes $U D_2 U^\dag$. Thus $\vN \{\cO\}$ contains all projections $|j\>\<j|$ and $|\td k\>\<\td k|$. It must therefore also contain
\be
|j\>\<j| \td k\>\<\td k |j'\>\<j'| = U_{jk} (U_{j'k})^* |j\>\<j'|,
\ee
and thus also contain $|j\>\<j'|$ for all $j$, $j'$ (as the matrix elements of $U$ are all nonzero). These operators obviously generate the entire matrix algebra $\cB(\bC^N)$.
\end{proof}

A special choice of a unitary matrix $U$ without any zero matrix elements is the discrete Fourier transform matrix:
\be
U_{jk} = \fr{e^{-2\pi i jk/N}}{\sqrt{N}},
\ee
although there are many other such matrices (and in fact they are generic among all unitary matrices).

We now return to presenting our example. Let $\cH_\UV = \cH_u \otimes \cH_\ub$ where $\cH_u$, $\cH_\ub$ have the same (arbitrary) dimension $N$. We choose the first seed state $|\y_1\>$ to be any state in $\cH_\UV$ whose reduced density operator $\check{\r}_u$ on $\cH_u$ has full rank. We choose the (unnormalized) second seed state to be
\be
|\y_2\> = \left[\left({\mathbbm 1}_u+\e \cO\right) \otimes{\mathbbm 1}_{\bar u}\right] |\y_1\>,
\ee
where $\e$ is any nonzero complex number and $\cO$ is an operator on $\cH_u$ whose matrix representation in some chosen orthonormal basis is given by $\cO=D_1 +i U D_2 U^\dag$ with matrices $D_1$, $D_2$, $U$ as in Lemma \re{alglemma}. Using
Lemma \re{aseedlemma} with $\r=|\y_2\>\<\y_1|$, $\check{\r}=|\y_1\>\<\y_1|$ tells us that $\cA^\seed_u$ contains
\be
\r_u \check{\r}_u^{-1} = (\mathbbm 1_u+\e \cO)\check{\r}_u \check{\r}_u^{-1} = \mathbbm 1_u+\e \cO.
\ee
According to Lemma \re{alglemma}, this operator must generate the entire $\cB(\cH_u)$, giving immediately
\be
\cA^\seed_u = \cB(\cH_u).
\ee

Acting with this $\cA^\seed_u$ on $\cH^\seed$, and noting that $\check{\r}_u$ having full rank means that $|\y_1\>$ is a so-called cyclic vector for $\cB(\cH_u)$, it is then a standard result that this action gives
\be
\cH^\pre_\IR = \cH_\UV,
\ee
and thus $\cH_\IR = \cH_\UV$ as well. Our construction did not find any common $|\c\>$ part among the seed states to integrate out, so it simply gives back the original $\cH_\UV$. Note that this is true even if we take the two seed states to be very close to each other (e.g., by making $\e$ very small) and even if $\cH_u$, $\cH_\ub$ have arbitrarily large dimensions.

%-------------------------------------------------
\section{Discussion}

\label{sec:disc}

Our work above provided an explicit construction taking as input a `UV' code with exact two-sided recovery and flat entanglement spectrum, together with a set of seed states.  The output of our construction was then an `IR' code with two-sided recovery and (to some given approximation) flat entanglement spectrum.  We think of this as modeling the action of bulk renormalization-group flow on holographic codes. However, in the limit in which the UV cutoff is removed one may also think of our RG flow as directly describing the relation between the full underlying dual CFT and the bulk effective field theory at an appropriate IR scale. Our construction identified the smallest possible subspace of $\cH_\UV$ which allows two-sided recovery and then `chopped' those states into pieces associated with small windows of eigenvalues for appropriately defined code modular Hamiltonians (see footnote~\re{codeK}), with the size of the windows controlling the approximation to which the resulting IR code has flat entanglement spectrum.

As advertised in the introduction, our construction allows non-trivial centers in both the UV and the IR.  This is more general than the framework considered in \cite{Gesteau:2023hbq}. Furthermore, in our construction the IR centers generally have no simple relation to those in the UV. A possible alternative construction is to bin the eigenvalues separately within each $\a$-sector of the UV code. In that case, one would obtain an RG flow where there is a natural inclusion of the UV central operators in the center of the IR algebra.

As noted previously, another important way in which our work differs from \cite{Gesteau:2023hbq} is that we provided an algorithmic construction of an IR code starting from the UV code given a set of states preserved by the RG flow, whereas \cite{Gesteau:2023hbq} assumed the existence of such a sequence of nested codes. Given such a sequence, one could, of course, use our idea of binning eigenvalues to construct a truly holographic code with an approximately flat entanglement spectrum. However, we also note that \cite{Gesteau:2023hbq} generalized their result to allow infinite-dimensional Hilbert spaces, whereas we worked with finite-dimensional Hilbert spaces for simplicity.  It would be useful to attempt to merge these approaches in the future to extend our results to infinite-dimensional cases.

It is natural to use our construction to help control analyses in contexts where small eigenvalues of reduced density matrices are problematic.  In particular, as pointed out in \cite{Kudler-Flam:2022jwd}, the JLMS relation can fail dramatically in such contexts.  Recall then that on a UV Hilbert space of dimension $D_\UV$, the smallest eigenvalue of a density matrix is necessarily of size $1/D_\UV$ or smaller, which can indeed be very small in the far UV.  For states that can be described by an IR Hilbert space of dimension $D_\IR \ll D_\UV$, the associated issues are then much reduced.  In general, this will reflect the fact that while small eigenvalues $\lambda$ may lead to large effects (perhaps through factors of $\ln \lambda$ or $\lambda^{-1}$), these effects can then be canceled by additional factors of $\lambda$ that may arise when a given eigenvalue is common to {\it all} states being studied (e.g., to all states that appear in $\cH_\IR$).

However, as illustrated by the second example in section \ref{sec:ex}, some choices of seed states can in fact preclude any reduction in the size of the code Hilbert space (no matter how large the eigenvalue windows are).  In such cases, our construction gives $\cH_\IR=\cH_{\UV}$ and thus offers little help in ameliorating the above concerns.   One may also be concerned that, even when it provides small code Hilbert spaces of small dimension (where a small code means a superselection sector in the IR), the construction given here can lead to a very large number of IR  superselection sectors (especially when the eigenvalue windows are taken to be very small so that each small code has very flat entanglement spectrum).  It seems likely that our construction can be further improved (in the sense of reducing the size of $\cH_\IR$ and, in particular, the number of IR superselection  sectors) by allowing additional errors in the IR FLM and R\'enyi-FLM formulas and, in particular, by allowing two-sided recovery (of the IR-to-UV map) to hold only in an approximate sense.  Indeed, we will show in \cite{ModFlow} that such improvements are quite useful when studying the exponentiated JLMS relation between bulk and boundary modular flows for semiclassical states.   Nevertheless, we leave for future work any  thorough study of this idea and of the associated trade-offs between code size and code precision. 

On the other hand, it is important to comment here on settings where we are given a UV code that satisfies two-sided reconstruction and flatness of the entanglement spectrum only up to some approximation.  In particular, in AdS/CFT we expect two-sided recovery to become imperfect at the order in the bulk Newton's constant $G$ where one can no longer neglect quantum fluctuations in the relevant quantum extremal surface.  In addition, taking the entanglement-spectrum of the UV code to be exactly flat appears to require projecting onto an exact eigenvalue of the area operator, which would then involve physics far outside the regime of semiclassical control; see related discussions in e.g.~\cite{Marolf:2018ldl,Dong:2019piw,Dong:2022ilf}.

Let us therefore suppose that the UV FLM formula \eqref{flmuv} and its R\'enyi analogue hold only up to some error bounds $\epsilon_{FLM}$ and $\epsilon_{FLM}(n)$. 
Recall that our construction proceeded entirely in the bulk, so that we were able to both define the IR code and to relate it to the UV code without relying on the UV (R\'enyi) FLM formulas for any R\'enyi index $n$.  Indeed, the final result was simply to equate the right-hand-side of the UV FLM formula with the right-hand-side of the IR FLM formula (and similarly for the R\'enyi generalizations up to an error set by the widths of our eigenvalue windows).  Repeating the arguments in the presence of non-zero $\epsilon_{FLM}$ and $\epsilon_{FLM}(n)$ then immediately gives an IR FLM formula with the same error bound $\epsilon_{FLM}$, as well as 
an IR R\'enyi FLM formula up to error bounds that are increased by $\epsilon_{FLM}(n)$ relative to those in section \ref{sec:general}. With this understanding, our analysis can be applied directly to the holographic bulk-to-boundary map at finite UV cutoff with finite-but-small bulk Newton's constant $G$.  We will make use of this as part of our study of modular flows in AdS/CFT in \cite{ModFlow}.

\section*{Acknowledgements}

This material is based upon work supported by the Air Force Office of Scientific Research under Award Number FA9550-19-1-0360. This material is also based upon work supported by the U.S. Department of Energy, Office of Science, Office of High Energy Physics, under Award Number DE-SC0011702. This work was supported in part by the Leinweber Institute for Theoretical Physics; and
by the Department of Energy, Office of Science, Office of High Energy Physics under
Award DE-SC0025293.

\appendix
\section{Explicit error bounds for the RG construction}
\label{app:explicit-bounds}

In this appendix, we give explicit bounds on the errors in the
approximation to the IR area operator and in the R\'enyi FLM formula derived in
section~\re{sec:general}. As in that section, we take the UV code
to have exact two-sided recovery and an exactly flat entanglement
spectrum. The same arguments apply to the single-seed construction
of section~\re{sec:onestate}.

Recall that the term linear in $n-1$ in~\er{renyitaylor}
vanishes by~\er{ngbdef} and~\er{airmexact}.
To make use of this cancellation without truncating the expansion,
it is useful to introduce, for real $t$,
\be\la{app:window-function}
g_\m(t):=\log\left[\sum_{\a,i}
\fr{|c_{\a i}^\b|^2}{|N_\g^\b|^2}
e^{t\(A_\IR^\m-\l_{\a i}^\b-\log|N_\g^\b|^2\)}\right],
\ee
where $\m=(\b,\g)$ labels a nonempty eigenvalue window and sums
over $(\a,i)$ at fixed $\m$ include only terms with
$\l_{\a i}^\b\in I_\g^\b$. Equations~\er{ngbdef}
and~\er{airmexact} give $g_\m(0)=g_\m'(0)=0$, while
$g_\m''(t)$ is the variance of $\l_{\a i}^\b$ with respect to
the normalized positive weights proportional to the summands
in~\er{app:window-function}. Since these eigenvalues lie in an
interval of width $\e_\m$, we have
\be\la{app:variance-bound}
0\leq g_\m''(t)\leq\fr{\e_\m^2}{4}.
\ee
To see the upper bound, note that the mean square distance from
the midpoint of the interval is at most $\e_\m^2/4$, while the
variance cannot exceed that mean square distance. Integrating
twice and using $g_\m(0)=g_\m'(0)=0$, we find
\be\la{app:integral-remainder}
g_\m(t)=\int_0^t(t-u)g_\m''(u)\,du
=t^2\int_0^1(1-v)g_\m''(vt)\,dv.
\ee
Here the second equality follows by setting $u=vt$, and holds
for either sign of $t$. Thus~\er{app:variance-bound} implies
\be\la{app:window-bounds}
0\leq g_\m(t)\leq\fr{t^2\e_\m^2}{8},\qqu t\in\bR.
\ee
The bound~\er{app:window-bounds} is often called Hoeffding's lemma.

\subsection{Approximation to the IR area operator}

We first apply~\er{app:window-bounds} to sharpen the
approximation~\er{airmdim}. Setting $t=-1$
in~\er{app:window-function} and using~\er{libdef}, we find
\ba
g_\m(-1)
&=\log\left[\sum_{\a,i}\fr{|c_{\a i}^\b|^2}{|N_\g^\b|^2}
e^{\l_{\a i}^\b+\log|N_\g^\b|^2}\right]-A_\IR^\m\nn\\
&=\log\left(\sum_\a e^{A_\UV^\a}D_\m^\a\right)-A_\IR^\m,
\ea
where $D_\m^\a$ is the number of terms in the $\a$-sector of
the window, as defined below~\er{ngbap}. We therefore obtain
\be\la{app:area-bound}
0\leq\log\left(\sum_\a e^{A_\UV^\a}D_\m^\a\right)-A_\IR^\m
\leq\fr{\e_\m^2}{8}.
\ee
This replaces the $O(\e_\m)$ error
in~\er{airmdim} by an explicit quadratic
bound and also fixes its sign. The approximation
$A_\IR^\m=\l_\g^\b+\log|N_\g^\b|^2+O(\e_\m)$ in~\er{airm}
still has a potentially linear error, since the midpoint
$\l_\g^\b$ need not equal the weighted mean of the eigenvalues
in the window.

\subsection{R\'enyi FLM formula for real $n$}

We now return to the derivation of~\er{renyieq2}, taking $n>0$
to be real. Using the definition of $g_\m$, we can
rewrite~\er{renyisectorexact} as
\be\la{app:sector-moment}
\sum_\a e^{(1-n)A_\UV^\a}\tr_o\(\hat\c_o^{\a\m}\)^n
=e^{(1-n)A_\IR^\m+g_\m(n-1)}.
\ee
Combining this result with~\er{renyifact} and the exact UV R\'enyi FLM
formula~\er{rflmuv2}, we may write the IR formula as
\be\la{app:exact-renyi}
e^{(1-n)S_n(\td\r_B)}
=\sum_\m p_\m^n e^{(1-n)\[A_\IR^\m+S_n(\r_r^\m)\]}
e^{g_\m(n-1)}.
\ee
Here $p_\m$ and $\r_r^\m$ are defined by~\er{rmdecom}, and
sectors with $p_\m=0$ may be omitted. If all
$\e_\m\leq\e$,~\er{app:window-bounds} gives
$1\leq e^{g_\m(n-1)}\leq e^{(n-1)^2\e^2/8}$.
Positivity of the summands then allows us to take the logarithm
and obtain an explicit error bound for~\er{rflmir2}:
\be
\left|S_n(\td\r_B)-\fr{1}{1-n}\log\left(
\sum_\m p_\m^n e^{(1-n)\[A_\IR^\m+S_n(\r_r^\m)\]}\right)\right|
\leq\fr{|n-1|\e^2}{8},\qu n>0,\quad n\neq1.
\la{app:renyi-bound}
\ee
The difference inside the absolute value is nonpositive for
$n>1$ and nonnegative for $0<n<1$. In the $n\to1$ limit,
the second term becomes $\<A_\IR\>_\r+S(\r_r)$, and the
difference vanishes exactly, reproducing~\er{flmir}.
The bound~\er{app:renyi-bound} holds without requiring
$|n-1|\e\ll1$, although it gives a small error only when its
right-hand side is small.

\subsection{R\'enyi FLM formula for $n=1+is$}

For $n=1+is$ with real $s$, we can again use $g_\m'(0)=0$
to bound the absolute error in~\er{rflmir2}. We define
complex powers using the real logarithm of each positive
eigenvalue, with zero eigenvalues contributing zero to the
powers considered here. Using $|e^{iy}-1-iy|\leq y^2/2$ for
real $y$, together with~\er{ngbdef} and~\er{airmexact},
we find
\be
\left|\sum_{\a,i}\fr{|c_{\a i}^\b|^2}{|N_\g^\b|^2}
e^{is\(A_\IR^\m-\l_{\a i}^\b-\log|N_\g^\b|^2\)}-1\right|
\leq\fr{s^2}{2}g_\m''(0)\leq\fr{s^2\e_\m^2}{8},
\la{app:complex-window-bound}
\ee
where the last step follows from~\er{app:variance-bound} at
$t=0$. The sum in~\er{app:complex-window-bound} is precisely the factor
multiplying $e^{-isA_\IR^\m}$ in~\er{renyisectorexact}
at $n=1+is$. Thus, using \er{rflmuv2}, \er{renyifact} and
$|\tr_r\(p_\m\r_r^\m\)^{1+is}|\leq p_\m$, we obtain for an
arbitrary IR state
\be
\left|\Tr_B\(\td\r_B\)^{1+is}
-\sum_\m e^{-isA_\IR^\m}\tr_r\(p_\m\r_r^\m\)^{1+is}\right|
\leq\fr{s^2}{8}\sum_\m p_\m\e_\m^2
\leq\fr{s^2\e^2}{8}.
\la{app:complex-bound}
\ee
Here the last inequality assumes $\e_\m\leq\e$ for all $\m$.
For a state in a single $\m$-sector, dividing by $|s|$ gives
\be\la{app:jlms-bound}
\fr{1}{|s|}\left|\Tr_B\(\td\r_B\)^{1+is}
-e^{-isA_\IR^\m}\tr_r\(\r_r^\m\)^{1+is}\right|
\leq\fr{|s|\e_\m^2}{8},\qqu s\neq0.
\ee
This is the convention of Definition~5 in~\cite{JLMS}: the
eigenvalue window contributes at most $|s|\e_\m^2/8$ to the
error bound $\xi_s\epsilon_{FLM}$, uniformly for all states
in the $\m$-sector. At $s=0$, the two traces agree by normalization;
in fact, both sides of~\er{app:jlms-bound} tend to zero as $s\to0$.

The absolute bound~\er{app:complex-bound} remains meaningful
when either term in the difference vanishes and requires no complex logarithm.
Any separately controlled absolute error in~\er{rflmuv2}
adds to the right-hand side of~\er{app:complex-bound} by the triangle inequality;
in~\er{app:jlms-bound}, that
additional error is divided by $|s|$ for $s\neq0$.

\bibliographystyle{JHEP}
\bibliography{bibliography}

\end{document}